\documentclass[a4paper,
               DIV=12,
               11pt,
               titlepage=off,
               abstract=true,
               listof=totoc]{scrartcl}
\pdfoutput=1
\usepackage{arxivstyle}
\usepackage{xcolor}
\usepackage[normalem]{ulem}

\DeclareMathOperator*{\Prb}{Pr}
\newcommand{\calL}{\mathcal L}
\newcommand{\calT}{\mathcal T}
\newcommand{\calU}{\mathcal U}
\newcommand{\calE}{\mathcal E}
\newcommand{\calV}{\mathcal V}
\newcommand{\supp}{\operatorname{supp}}
\newcommand{\wt}{\operatorname{wt}}
\newcommand{\spanop}{\operatorname{span}}
\newcommand{\polylog}{\operatorname{polylog}}
\newcommand{\ot}{\widetilde O}

\newcommand{\Params}[3]{{\textsf{Parameters}\left\{
\begin{array}{r  l}
\mbox{\scriptsize\sffamily queries to $U_\OO$:}&\enspace {#1} \\
\mbox{\scriptsize\sffamily queries to $U_M$:}&\enspace {#2} \\
\mbox{\scriptsize\sffamily time:}&\enspace {#3} \\
\end{array}
\right.}}
\newcommand{\Yes}{\textnormal{\textsc{Yes}}\xspace}
\newcommand{\No}{\textnormal{\textsc{No}}\xspace}
\newcommand{\Fail}{\textnormal{\textsc{Fail}}\xspace}
\newcommand{\Select}{\mathsf{Select}}

\newcommand{\Extract}{\mathsf{Extract}}
\newcommand{\Bog}{\mathsf{BR}}
\newcommand{\Red}{\mathcal{R}}

\newcommand{\Nodes}{\operatorname{Nodes}}
\newcommand{\Kext}{\mathbb K}

\title{Improved Quantum Random Self-Reduction for Linear Problems}

\author[1]{Vahid R.\ Asadi\thanks{\href{mailto:vrasadi@nii.ac.jp}{\texttt{vrasadi@nii.ac.jp}}}}
\author[1]{Shuichi Hirahara\thanks{\href{mailto:s_hirahara@nii.ac.jp}{\texttt{s\_hirahara@nii.ac.jp}}}}
\author[2]{Nobutaka Shimizu\thanks{\href{mailto:shimizu.n.ah@m.titech.ac.jp}{\texttt{shimizu.n.ah@m.titech.ac.jp}}}}

\affil[1]{\textit{National Institute of Informatics}}
\affil[2]{\textit{Institute of Science Tokyo}}

\date{}
\hypersetup{pdftitle={Improved Quantum Random Self-Reduction for Linear Problems}}
\begin{document}
\maketitle

\begin{abstract}
        We study quantum random self-reductions for linear problems over finite fields.  Let $M\in\mathbb{F}^{n\times n}$ be an arbitrary matrix, and let $\mathcal{O}$ be an oracle that agrees with the linear map $x\mapsto Mx$ on an $\varepsilon$-fraction of inputs $x\sim\mathbb{F}^n$.  Given coherent access to $\mathcal{O}$ and coherent entry access to $M$, we give a uniform quantum reduction that computes $Mx$ on any prescribed input $x$ with probability at least $2/3$ in time $\widetilde{O}(nT^{1/3})$, for $n\le T\le n^{3/2}$ and constant field size and $\eps$, where $T$ is the cost of one coherent query to $\mathcal{O}$.  In particular, when $T=\widetilde{O}(n)$, the reduction runs in time $\widetilde{O}(n^{4/3})$, improving the $\widetilde{O}(n^{3/2}+T)$ reduction of Asadi, Golovnev, Gur, Shinkar, and Subramanian (SODA 2024).
        
        Our reduction uses the Bogolyubov--Ruzsa subspace guaranteed by additive combinatorics, but it avoids learning this subspace explicitly, which was computationally expensive for the previous reduction; in particular, it does not recover a basis for its orthogonal complement. The main technical step is to decompose the inputs into sparse pieces and find a vector that lies outside the Bogolyubov--Ruzsa subspace via a quantum search based on amplitude amplification. This yields a tunable tradeoff between the cost of querying the average-case oracle and the cost of verifying matrix-vector products.
\end{abstract}

\newpage

\tableofcontents

\newpage


\section{Introduction}
The \emph{linear problem} $\mathcal{L}_M$ specified by a matrix $M\in\mathbb{F}^{n\times n}$ over a finite field $\mathbb{F}$ is the following basic problem: Given an input vector $x\in\mathbb F^n$, compute the product $Mx\in\mathbb F^n$. This basic problem forms a broad and fundamental class of computational tasks such as the discrete Fourier transform and multipoint evaluation of polynomials. Although every linear problem can be solved by the naive $O(n^2)$-time algorithm, many structured matrices admit nearly linear-time algorithms.

In applications, the relevant complexity measure is not necessarily the worst-case complexity of computing $Mx$ on every input $x$, but rather the average-case complexity where the input $x\sim\mathbb{F}^n$ is a uniformly random vector.
For example, Wiedemann-type Krylov methods \cite{Kal91} access a matrix through the map $x\mapsto Mx$, start from a random vector $x\sim\mathbb F^n$, and use the sequence $x,Mx,M^2x,\ldots$ to solve tasks such as linear-system solving, minimal-polynomial computation, and determinant computation.

This motivates the \emph{random self-reduction} for linear problems: if an oracle $\mathcal{O}$ computes $Mx$ correctly on an $\varepsilon$-fraction of inputs, i.e.,
\begin{align}
\Pr_{z\sim\mathbb F^n}[\mathcal O(z)=Mz]\ge \varepsilon, \label{eq:oracle success prob}
\end{align}
can we design an oracle algorithm $A^\mathcal{O}$ that computes $Mx$ for every input $x$ without performing the naive $O(n^2)$-time multiplication?

When $\varepsilon>3/4$, a standard self-correction already gives a simple worst-case-to-average-case reduction~\citep{BLR93}: On input $x$, sample a random vector $r\sim\mathbb{F}^n$ and compute $\mathcal{O}(r)+\mathcal{O}(x-r)$, which equals $Mx$ with probability $1/2+2c$ (over the choice of $r$) if $\varepsilon\ge 3/4+c$. Repeating this for $O(1/c^2)$ times and taking majority vote yields $Mx$ with high probability.

Thus, the main difficulty lies in the low-agreement regime where $\varepsilon$ is small. To state it more formally, we say that a worst-case-to-average-case reduction is \emph{$\varepsilon$-tolerant} if, for every oracle $\mathcal{O}$ satisfying \cref{eq:oracle success prob}, it computes $Mx$ with probability $2/3$ (over the internal randomness of the reduction) on every input $x\in\mathbb{F}^n$.

The first general progress dealing with such tolerant reductions was made through the additive-combinatorics framework of \citet{AGGS22}. Their key insight was that, even if the set of inputs on which $\mathcal{O}$ is correct is completely unstructured, additive combinatorics can still extract a large hidden subspace, called a Bogolyubov--Ruzsa (BR) subspace, from it. This makes it possible to design tolerant reductions for several settings, including data structures for all linear problems. The main bottlenecks to applying this framework directly to linear problems are learning the BR subspace and verification, which appear to require learning large Fourier-coefficients~\citep{GL89} and computing $Mx$ itself in the classical oracle model.

A subsequent work by \citet{AGGSS24} overcame these bottlenecks in the quantum setting, assuming oracle access to the entries of $M$ in addition to oracle access to $\mathcal O$.\footnote{Some access to $M$ is necessary in general. For example, let $\mathcal M=\{M_m:m\in\mathbb F^n\setminus{0}\}$, where $M_m$ is the matrix whose first row is $m$ and whose remaining entries are zero. Let $\mathcal O$ be the oracle that always outputs the all-zero vector. Then, for every $M_m\in\mathcal M$, we have $\mathcal O(x)=M_mx$ on a $1/|\mathbb F|$ fraction of inputs $x$, namely those satisfying $\langle m,x\rangle=0$. Thus, from oracle access to $\mathcal O$ alone, the reduction cannot distinguish the possible matrices $M_m$, even though the desired value $M_mx$ depends on $m$ for a worst-case input $x$.} They provide quantum subroutines for both tasks needed to instantiate the additive-combinatorics framework for all linear problems. Consequently, they proved that every linear problem admits a quantum $\varepsilon$-tolerant reduction whose running time is roughly $\widetilde{O}_\varepsilon(T+n^{3/2})$, where $T$ is the time needed to evaluate $\mathcal O$.

\citet{HS26} later presented the optimal classical reduction for linear problems in the nonuniform circuit setting that runs in time $\widetilde{O}_{\varepsilon}(n+T)$. At the same time, they proved a lower bound for uniform classical reductions: over small fields, any classical subquadratic-time random self-reduction requires an advice string of length $\Omega(\log(1/\varepsilon)\log n)$. An advice string means auxiliary information, depending on the matrix and the average-case oracle, that is supplied to the reduction rather than computed by the reduction itself.

Together, these results reveal a sharp contrast between classical and quantum reductions. On the classical side, the result of \citet{HS26} shows that any subquadratic-time uniform reduction must receive nontrivial advice. On the quantum side, the reduction of \citet{AGGSS24} is uniform, i.e., requires no advice, and runs in time $\widetilde O_\varepsilon(T+n^{3/2})$. This leaves the following natural quantitative question: 
\begin{quote}
\emph{Once quantum computation is allowed, what is the best uniform reduction one can obtain?} In particular,
\emph{if the average-case oracle is fast, can one improve on the \(n^{3/2}\)-time
overhead of \cite{AGGSS24}?}
\end{quote}

We answer this question by giving a tunable quantum reduction. To state the result informally, assume in
this introduction that the field size and \(\eps\) are both universal constants.
We refer to \cref{sec:quantum-computing} for the model of quantum computation.

\begin{theorem}[Informal main theorem; see \cref{thm:main,cor:optimized}]
\label{thm:introduction-informal}
Let $\mathcal{O}$ be any oracle satisfying \cref{eq:oracle success prob}. Let one quantum query to the average-case oracle $\OO$ cost time $T$, where $n\le T\le n^{3/2}$.
There is a uniform quantum reduction which, given $x\in\F^n$, outputs $Mx$ with probability at least $2/3$, 
and runs in time $\ot(nT^{1/3})$. In particular, when $T=\Theta(n)$, the running time becomes $\ot(n^{4/3})$.
\end{theorem}

The improvement of \cref{thm:introduction-informal} over \cite{AGGSS24} is gradual: for $T=n$ the running time is $\ot(n^{4/3})$, and as $T$ approaches $n^{3/2}$, the bound becomes $\ot(n^{3/2})$, matching the overhead of \cite{AGGSS24}.  A notable feature of our approach is that it does not explicitly learn the Bogolyubov--Ruzsa (BR) subspace, whose recovery is one of the main algorithmic tasks in the additive-combinatorics framework of \cite{AGGS22,AGGSS24}.  In particular, our reduction avoids the quantum Fourier-sampling step used in \cite{AGGSS24} for this purpose. 
Interestingly, we still use the existence of the BR subspace, but never recover a basis or a membership oracle for it.
Thus, the improved bound in \cref{thm:introduction-informal} can be viewed as arising from a more efficient way of finding and exploiting verification failures, implemented via a quantum search over a sparse consistency tree.

The extension to quantum average-case oracles also gives a consequence for quantum algorithms for linear problems.  Suppose a quantum algorithm uses $T\ge n$ gates and outputs $Mz$ with probability at least $\eps$ on a uniformly random input $z\in\F^n$, where the probability also includes its internal measurements.  For constant field size and $\eps$, it can be converted into a quantum algorithm that outputs $Mx$ with probability at least $2/3$ for every input $x$, using
\begin{align*}
        \ot\left(T+nT^{1/3}\right)
        \enspace
\end{align*}
gates. In particular, when $T=\ot(n)$, the worst-case gate complexity is $\ot(n^{4/3})$.  We give the formal statement in \cref{cor:quantum-algorithms}.

‌Before proving our main result and as a warm-up, we give a uniform classical reduction in the error-less model, where $\OO(z)$ is always either $Mz$ or $\bot$, and it is not $\bot$ on at least an $\eps$-fraction of inputs (see \cref{def:errorless}).
Note that the error-less model differs from the general error-prone guarantee of \cref{eq:oracle success prob} in that the oracle is not allowed to output the wrong answer, which in particular makes the task of verification trivial. This reduction uses a simple recursive repair procedure and runs in time $\ot(T)$, for constant field size and $\eps$.  In particular, when $T=O(n)$, its running time is nearly linear.  An error-prone oracle together with a perfect verifier can be converted to an error-less oracle by replacing every rejected answer by $\bot$. Therefore, as a corollary, using the quantum matrix-vector verifier of \cite{AGGSS24}, we can recover their $\ot(T+n^{3/2})$ running-time bound without Fourier sampling or learning the BR subspace. Conceptually, this indicates that the main \emph{quantum} speedup achieved by \cite{AGGSS24} is due to fast verification, and not Fourier sampling.  We describe the error-less model and the repair procedure in \cref{sec:warm-up-errorless}.

\begin{figure}[t]
\small
\centering
\begin{tikzpicture}[x=1cm,y=1cm,>=stealth,font=\footnotesize]
  \def\xN{1.35}
  \def\xH{5.65}
  \def\xQ{9.95}
  \def\yN{0.82}
  \def\yF{2.02}
  \def\yH{2.62}
  \def\yQ{4.42}

  \draw[gray!18] (\xN,\yN) rectangle (\xQ,\yQ);
  \foreach \x in {\xN,\xH,\xQ}
    \draw[gray!18] (\x,\yN) -- (\x,\yQ);
  \foreach \y in {\yN,\yF,\yH,\yQ}
    \draw[gray!18] (\xN,\y) -- (\xQ,\y);
  \draw[->,thick] (\xN,\yN) -- (10.35,\yN);
  \draw[->,thick] (\xN,\yN) -- (\xN,4.70);

  \node[anchor=north] at (5.65,0.18) {average-case oracle cost $T$};
  \node[rotate=90,anchor=south] at (0.18,2.62) {total running time};

  \foreach \x/\lab in {\xN/$n$,\xH/$n^{3/2}$,\xQ/$n^2$}
    \draw[black!55] (\x,\yN-0.055) -- (\x,\yN+0.055)
      node[below=3pt,black,font=\scriptsize] {\lab};
  \foreach \y/\lab in {\yN/$n$,\yF/$n^{4/3}$,\yH/$n^{3/2}$,\yQ/$n^2$}
    \draw[black!55] (\xN-0.055,\y) -- (\xN+0.055,\y)
      node[left=4pt,black,font=\scriptsize] {\lab};

  \draw[red!70!black,very thick,dashed,line cap=round]
    (\xN,\yH) -- (\xH,\yH);
  \draw[blue!70!black,very thick,line cap=round]
    (\xN,\yF) -- (\xH,\yH);
  \draw[black!55,very thick,line cap=round]
    (\xH,\yH) -- (\xQ,\yQ);

  \node[anchor=west,font=\scriptsize,red!70!black]
    at (1.82,3.54) {prior work \cite{AGGSS24}: $\ot(T+n^{3/2})$};
  \node[anchor=west,font=\scriptsize,blue!70!black]
    at (1.95,1.65) {this work: $\ot(nT^{1/3})$};
  \node[sloped,above=6pt,font=\scriptsize,black!60]
    at ($(\xH,\yH)!0.58!(\xQ,\yQ)$) {$T\ge n^{3/2}$: $\ot(T)$};

  \fill[blue!70!black] (\xN,\yF) circle (1.4pt);
  \fill[black!70] (\xH,\yH) circle (1.4pt);
  \fill[black!70] (\xQ,\yQ) circle (1.4pt);
\end{tikzpicture}
\caption{Running-time comparison in the constant-field, constant-agreement regime, suppressing polylogarithmic factors.  Both axes are drawn in log-log scale from $n$ to $n^2$, with the right and top endpoints marking the trivial quadratic-time scale.  The improvement occurs for $n\le T\le n^{3/2}$, where the optimized bound in this work goes from $\widetilde O(n^{4/3})$ at $T=n$ to $\widetilde O(n^{3/2})$ at $T=n^{3/2}$.  For larger $T$, both bounds are at the $\widetilde O(T)$ scale up to $T=n^2$.}
\label{fig:intro-runtime-tradeoff}
\end{figure}
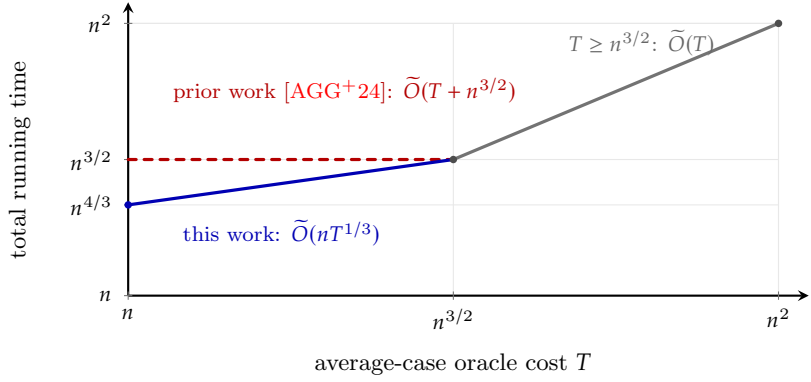

\subsection{Previous work}
\label{sec:previous-work}
Random self-reductions, or worst-case-to-average-case reductions, have a
long history, going back in particular to the seminal
work of Blum, Luby, and Rubinfeld \cite{BLR93}.  More recently, Asadi,
Golovnev, Gur, and Shinkar \cite{AGGS22} introduced an additive-combinatorial
framework for designing such reductions, and applied it to a variety of problems, most notably matrix multiplication.  A key idea in this framework is
that, even if the set of inputs on which the average-case algorithm is correct
has no visible structure, a Bogolyubov--Ruzsa type lemma gives a hidden
structured set on which one can self-correct.

The closest prior quantum work is due to Asadi, Golovnev, Gur, Shinkar, and
Subramanian \cite{AGGSS24}.  They showed that every linear problem admits a
uniform quantum worst-case-to-average-case reduction.  Their reduction uses
quantum singular value transformation to implement verification in
superposition, and quantum Fourier sampling ideas to learn the hidden
Bogolyubov--Ruzsa subspace.  Quantitatively, their reduction has an
\(n^{3/2}\)-time overhead.

The closest prior classical work is due to Hirahara and Shimizu \cite{HS26}, following up on  a series of works \cite{HS22, HS23}.
They gave optimal error-tolerant random self-reductions for all linear
problems with advice, and proved that, over fields of size $q$, classical
subquadratic-time random self-reductions require advice of size
\(\Omega(\log_q(1/\eps)\cdot\log n)\).  Thus, the uniform quantum setting is
qualitatively different from the uniform classical one.

In addition, a result of Aggarwal and Kwan \cite{AK25} simplifies the quantum reduction for linear problems using techniques from \cite{HS23}. However, 
their setting is different from ours: In their average-case problem, both the matrix $M$ and the input vector $x$ are random, whereas in our setting, only the input vector $x$ is random. This distinction appears to be essential for their use of techniques from \cite{HS23}.
Moreover, as their result relies on verification of matrix-vector products, in particular, that of \cite{AGGSS24}, it will not result in an improvement on the overhead.

\paragraph{Organization.}
{We give a detailed overview of our results in \cref{sec:overview}, starting with the error-less reduction and its verifier-based corollary.} \cref{sec:prelim} introduces the oracle model, and the necessary tools and ingredients for this result.  \cref{sec:tree} defines the consistency tree and proves
the basic structural facts about what we define as bad nodes.  \cref{sec:extraction} gives the sparse
extraction procedure.  Finally, \cref{sec:full-algorithm} proves the main theorem and
derives the optimized parameter choices. 

\paragraph{AI Disclosure.}
{The algorithm and proof strategy underlying the main theorem were suggested by
ChatGPT 5.5. Subsequent interactions with AI were used to refine and clarify the proof.
The authors independently verified, developed, and simplified the argument presented here and take full
responsibility for its correctness, exposition, and attribution.}

\paragraph{Acknowledgements.}
{Vahid R.\ Asadi is supported by a JSPS Postdoctoral Fellowship for Research in Japan, and by JSPS KAKENHI Grant Number 26KF0111.}
Shuichi Hirahara is supported by JSPS KAKENHI Grant Number 24K21317.
Nobutaka Shimizu is supported by JSPS KAKENHI Grant Numbers 24K21317, 23K16837 and JST PRESTO Grant Number JPMJPR26K4, Japan.

\section{Proof overview}
\label{sec:overview}

We now explain the main ideas of the reduction while suppressing some details. For simplicity, we assume $\F=\F_2$ throughout this overview.  Suppose the input to the problem is an oracle $\OO:\F^n\to\F^n$, a matrix $M\in\F^{n\times n}$, and an input $x\in\F^n$.  The goal is to compute $Mx$, given that the oracle $\OO(v)$ agrees with $Mv$ on an $\eps$-fraction of inputs $v$.  We denote the running time of the oracle by $T$.  We assume that $\OO$ and $M$ can be queried in superposition, meaning that we have access to unitaries $U_\OO$ and $U_M$, where $U_M$ is the entry oracle for $M$.  See \cref{sec:prelim} for the formal model.  Let
\begin{align*}
        A=\{v\in\F^n:\OO(v)=Mv\}
        \enspace
\end{align*}
be the set of \emph{good} inputs, namely those that the oracle answers correctly.  By the average-case assumption, the set $A$ has density at least $\eps$.  However, $A$ need not have any useful linear structure.  Nevertheless, the Bogolyubov--Ruzsa lemma lets us extract a hidden subspace from it. 
Informally, it states that there is a subspace $V\subseteq\F^n$ of codimension $O(\log(1/\eps))$ such that every $v\in V$ has many representations of the form $v=a_1+\dots+a_t-b_1-\dots-b_t$ with $a_1,\dots,a_t,b_1,\dots,b_t\in A$, where $t=O(\log(1/\eps))$. 
To state it more formally, define $tA-tA=\{a_1+\dots+a_t-b_1-\dots-b_t\colon a_1,\dots,a_t,b_1,\dots,b_t\in A\}$.

\begin{lemma}[Informal; see \cref{thm:BR}]\label{lem:BR-informal}
For any set $A\subseteq\F^n$ of density $\eps$, and for $t=\Omega(\log(1/\eps))$, there exists a subspace $V\subseteq\F^n$ of codimension $O(\log(1/\eps))$ such that $V\subseteq tA-tA$.  Moreover, for every $v\in V$, a random $tA-tA$ representation succeeds with probability at least $\eps^{2t+1}$.
\end{lemma}

Thus, if $v\in V$ and we sample enough BR representations of $v$, one of the corresponding oracle sums is equal to $Mv$ with high probability.
We refer to the list of oracle sums obtained from these sampled representations as the \emph{BR candidate list} for $v$.

Moving from the unstructured set $A$ to a self-correctable subspace $V$ is a good first step.  However, a barrier is that the algorithm does not know $V$.  In \cite{AGGSS24}, one of the main ideas is to recover a basis for $V^\perp$ using quantum Fourier sampling techniques.  Here we avoid this, as it requires a procedure to evaluate the indicator function $1_A$, which in turn requires verification for arbitrary inputs.  Instead, we increase the self-correctable space by maintaining a certified list $G$ of pairs $(g,Mg)$ and working with the larger hidden space
\begin{align*}
        W_G:=V+\spanop\{g:(g,Mg)\in G\}
        \enspace .
\end{align*}
The important observation is that for every $v\in W_G$, we can brute-force over all possible coefficients for the elements of $G$ to find a decomposition $v=v'+\alpha_1g_1+\cdots+\alpha_kg_k$ with $v'\in V$.  Once we have such a $v'$, the BR candidate list for $v'$ contains $Mv'$ with high probability over the sampled randomness, and therefore gives a way to build a candidate list for $Mv$ as well.

So far, we have only said that $Mv$ belongs to a small list.  We also need a way to recover the correct value from that list.  This is done by the tournament selection procedure of \cite{HS26}.

\begin{lemma}[Informal; see \cref{lem:select}]\label{lem:tournament-informal}
Given $v\in\F^n$ and an explicit list of candidates $(y^{(1)},\ldots,y^{(K)})$ for $Mv$ that is guaranteed to contain the correct answer, one can find $Mv$ using at most $Kn$ queries to $M$.
\end{lemma}

The idea is simple: keep one candidate as the current winner, compare it with the next candidate on a coordinate where they differ, and read the corresponding row of $M$ to see which candidate agrees with $Mv$ on that coordinate.  If $Mv$ is in the list, it can never be eliminated, so tournament selection eventually finds it.

We have reduced the problem to the following task: find certified pairs $(g,Mg)$ with $g\notin W_G$, and add them to $G$ so that the dimension of $W_G$ increases by one.  Though one has to be careful here.  A vector outside $W_G$ may still accidentally have a correct candidate in a BR representation list.  Thus, we cannot use BR self-correction and verification as a membership test for $W_G$.  Instead, we only use the one-sided implication that if verification fails, then (except on the event of sampling bad randomness for BR representations) the vector is outside $W_G$.

In the following subsections, we explain how to find such pairs starting from the requested input $x$, and how the final algorithm will look.  {But, before going into the main reduction, we first take a short detour and give a warm-up in the error-less model.  This isolates the recursive repair argument that will be used later in the main result.  The verifier-based reduction matching the running-time bound of \cite{AGGSS24} then follows as a corollary.}

\subsection{{Warm-up: correction in the error-less model}}
\label{sec:warm-up-errorless}

Before proving the main theorem, we describe a simpler version of the argument in the error-less oracle model. Here, the oracle is not allowed to output a wrong answer, but has to tell us when it cannot answer by outputting a special symbol $\bot$. Thus, the oracle answers correctly on at least an $\eps$-fraction of inputs, and outputs $\bot$ on the rest. This lets us repair the current self-correction space by a purely classical method.

\begin{definition}[Error-less oracle]\label{def:errorless}
An oracle $\OO:\F^n\to\F^n\cup\{\bot\}$ is \emph{error-less} for $M$ with agreement at least $\eps$ if both of the following hold:
\begin{itemize}
    \item $\OO(z)\in\{Mz,\bot\}$ for every $z\in\F^n$
    \item $\Pr_{z\sim\F^n}[\OO(z)=Mz]\ge\eps$
\end{itemize}
\end{definition}

As before, let $A=\{z:\OO(z)=Mz\}$ be the set of \emph{good} inputs, and let $V$ be the hidden BR subspace from \cref{lem:BR-informal}.  Also let $G$ be a certified list of pairs of the form $(g,Mg)$, and recall that $W_G:=V+\spanop\{g:(g,h)\in G\}$.

Let $R\in\mathbb{N}$ be a parameter chosen large enough for the required good randomness events, and let $\Omega$ be a BR randomness table of $R$ trials.  We form the BR candidate list as before, except that we discard a candidate whenever one of its oracle answers is $\bot$.  If a candidate returns a non-$\bot$ answer, then it must be the correct answer by the error-less guarantee of the oracle (so in this setting, we do not need to run the tournament selection).  Return any such candidate, or $\bot$ if none remains, and denote this procedure by $\mathsf{Correct}_{G,\Omega}(v)$.  Over $\F_2$, we try roughly $2^{|G|}R$ candidates.

The point is the following: if $v\in W_G$, then, with high probability over $\Omega$, one of the attempted BR representations uses only inputs in $A$.  In that case, $\mathsf{Correct}_{G,\Omega}(v)=Mv$.  Thus, getting a $\bot$ answer certifies that $v\notin W_G$, except on the list-failure event.

We now describe the repair step.  Our goal here is to increase the dimension of $W_G$, by finding a pair $(g,Mg)$ where $g$ is not in $W_G$. It starts from a vector $v$ whose correction returned $\bot$, and therefore, on the good-list event, $v$ lies outside the current $W_G$.  The procedure descends through the support of $v$ recursively, by writing it as $v=v_0+v_1$, where each child has roughly half the support of $v$.  It may reach a weight-one vector, but it may also stop earlier if both children are corrected successfully.

\begin{tcolorbox}[title=Recursive subroutine: \textsc{Repair}]
\paragraph{Input:} Oracles $\OO$ and $M$, a certified list $G$, and a vector $v\in\F^n$ such that $v\notin W_G$.

\paragraph{Goal:} Add one new certified pair $(g,Mg)$ with $g\notin W_G$ to $G$.

\begin{itemize}
    \item If $\wt(v)\le1$, compute $Mv$ directly from $M$, update $G\gets G\cup\{(v,Mv)\}$ and return.

    \item Otherwise, split $v$ into two vectors $v=v_0+v_1$, each supported on roughly half of $\supp(v)$.

    \item Sample fresh BR randomness tables $\Omega_0$ and $\Omega_1$ for the two children, and compute $y_i=\mathsf{Correct}_{G,\Omega_i}(v_i)$ for $i\in\{0,1\}$.

    \begin{itemize}
        \item If $y_0=\bot$, call $\textsc{Repair}^{\OO,M}(G,v_0)$.

        \item Otherwise, if $y_1=\bot$, call $\textsc{Repair}^{\OO,M}(G,v_1)$.

        \item Otherwise, both children have been corrected.  Update $G\gets G\cup\{(v,y_0+y_1)\}$.
    \end{itemize}
\end{itemize}
\end{tcolorbox}

The logic of the procedure is the following:  If a child $v_i$ returns $\bot$, then the good-list event implies $v_i\notin W_G$, and we may recurse on it.  If both children return answers, then $y_0=Mv_0$ and $y_1=Mv_1$, so $y_0+y_1=Mv$.  Since the input to the procedure satisfies $v\notin W_G$, appending $(v,y_0+y_1)$ is certified and gives progress.

The full warm-up algorithm repeatedly tries to self-correct the requested input $x$.  If correction succeeds, we are done.  If it returns $\bot$, we run the repair procedure starting from $x$ and thereby enlarge $W_G$ by one dimension.

\begin{tcolorbox}[title=Warm-up algorithm]
\paragraph{Input:} Oracles $\OO$ and $M$, and an input vector $x\in\F^n$.

\paragraph{Output:} $y=Mx$.

\begin{itemize}
    \item Initialize $G\gets\emptyset$.

    \item Repeat for $O(\log(1/\eps))$ rounds.
    \begin{itemize}
        \item Sample a fresh BR randomness table $\Omega$ for the current list $G$.

        \item Compute $y=\mathsf{Correct}_{G,\Omega}(x)$.

        \item If $y\ne\bot$, output $y$ and return.

        \item Otherwise, call $\textsc{Repair}^{\OO,M}(G,x)$.
    \end{itemize}
\end{itemize}
\end{tcolorbox}

\begin{theorem}[Error-less warm-up]\label{thm:warm-up-errorless}
Let $0<\eps\le1/2$.  Given an error-less oracle with agreement at least $\eps$, the warm-up algorithm outputs $Mx$ for any given $x$ with probability at least $2/3$ over the internal randomness.  It makes $\ot((2/\eps)^{O(\log(1/\eps))})$ calls to $\OO$ and $O(n\log(1/\eps))$ entry queries to $M$.  If one call to $\OO$ costs time $T\ge n$, then the running time is
\begin{align*}
        \ot\left((2/\eps)^{O(\log(1/\eps))}T\right)
        \enspace .
\end{align*}
\end{theorem}

For constant $\eps$, there are only $O(\log n)$ correction calls along the repair paths.  We leave the sample sizes and the proof to \cref{apx:proof}, as the main point here is the recursive correction argument.

Note that an error-prone oracle together with a perfect verifier gives an error-less oracle: evaluate $\OO(z)$, return its answer if the verifier accepts, and return $\bot$ otherwise.  We therefore obtain the following verifier-based reduction as a corollary.

\begin{corollary}[Correction with a verifier]\label{thm:warm-up-verifier}
Suppose $\OO$ has agreement at least $\eps$, and a perfect verifier $\calV_M$ decides whether a proposed answer is correct.  If one call to $\OO$ costs time $T\ge n$ and one call to $\calV_M$ costs time $S$, then the warm-up gives a reduction with running time
\begin{align*}
        \ot\left((2/\eps)^{O(\log(1/\eps))}(T+S)\right)
        \enspace .
\end{align*}
Using the quantum matrix-vector verifier of \cite{AGGSS24}, gives running time $\ot((2/\eps)^{O(\log(1/\eps))}(T+n^{3/2}))$.
\end{corollary}

Thus, we recover the running-time bound of \cite{AGGSS24} without Fourier sampling or learning the BR subspace.  The proof, including the error bound for the quantum verifier, is given in \cref{apx:proof}.

\subsection{A leaf-only idea and why it fails}

{We now return to the general case of the error-prone model and the extra ideas needed for the main result.  The verifier-based corollary of the warm-up checks oracle answers on arbitrary inputs. Therefore, using the matrix-vector verification algorithm of \cite{AGGSS24} costs about $n^{3/2}$ per check, which is too expensive for the tradeoff we want. }

A first attempt to avoid this blowup is to use the following intuition.  Verifying matrix-vector multiplication is cheaper when the input vector is sparse, as we only need to work with a small subset of columns of $M$.  Therefore, we can split the input $x$ into sparse blocks $x=v_1+\cdots+v_m$, where each $v_i$ has Hamming weight at most $b$, for some $b\in[n]$, and $m=O(n/b)$.

For each block $v_i$, compute a selected candidate $y_i$ and verify whether $y_i=Mv_i$ using the matrix-vector verification procedure of \cite{AGGSS24}.  If some block fails verification, then we have found a vector outside $W_G$, and we can increase $G$ and make progress toward making $x$ lie in $W_G$.  We can search for such a bad block quantumly using amplitude amplification.  However, an issue arises when \emph{no block} verification fails.  In this case, every selected block value is correct, so
\begin{align*}
        Mx=\sum_i y_i
        \enspace .
\end{align*}
This is logically sound, but too expensive to compute.  The quantum search tells us that no bad block exists, but it does not materialize all the values $y_i$.  To output the sum above, we would still have to compute all $O(n/b)$ selected block values, giving an extra $\ot(Tn/b)$ cost.

\subsection{The consistency tree}

The consistency tree is a way to replace the expensive leaf sum by one root computation plus local consistency checks.  Imagine the full binary tree over the coordinates of $x$, and cut it at the level where the nodes have weight at most $b$.  These cut-level vectors are the leaves of the working tree.  For each node $u$, let $v_u$ be the sum of the leaf vectors below $u$, and let $y_u$ be the outcome of tournament selection on $\calL_{G,\Omega}(v_u)$.  The value $y_u$ is the algorithm's current selected candidate for $Mv_u$.

A leaf $u$ is bad if $y_u\ne Mv_u$.  This can be checked by restricted verification, since leaves have weight at most $b$.  An internal node $u$ with children $u0,u1$ is bad if
\begin{align*}
        y_u\ne y_{u0}+y_{u1}
        \enspace .
\end{align*}
Note that this is a local additive check as it only checks consistency with the children.  If the whole tree has no bad node, then a bottom-up induction gives $y_u=Mv_u$ at every node, and in particular the root value is $Mx$.  Thus the algorithm can output one selected root value instead of explicitly summing all leaves.

If a bad node exists, we can find one by a quantum search over the nodes of the tree.  We then descend to a minimal bad node, one that does not have a bad node among its proper descendants, by searching the two child subtrees and recursively moving into a subtree that contains a bad node.  Since the searched subtrees shrink geometrically, this descent has the same asymptotic cost as the first search.

So far, we have shown that the tree lets us certify the correctness of the root value if no node is bad, and also lets us find a minimal bad node if one exists.  The next step is to explain how to make progress having found a minimal bad node.

\subsection{Repairing a bad vector}

A minimal bad internal node is easy to repair.  If $u$ is internal and minimal, then its children are not bad, so their selected values are correct.  Hence
\begin{align*}
        Mv_u=Mv_{u0}+Mv_{u1}=y_{u0}+y_{u1}
        \enspace .
\end{align*}
Since $u$ itself is bad, the vector $v_u$ cannot lie in the current $W_G$.  Therefore appending $(v_u,y_{u0}+y_{u1})$ to $G$ increases the dimension of $W_G$.

If the minimal bad node is a leaf, then its vector has weight at most $b$.  We repair it by a classical binary descent, similar to the one in the \textsc{Repair} procedure.  Split the current bad vector $v$ as $v=v_0+v_1$.  Compute selected candidates for $v_0$ and $v_1$ and verify them.  If one child fails verification, recurse on that child.  If both children verify, then their selected values are correct, so their sum is $Mv$, and we append $(v,Mv)$ to $G$.  The key observation is that all vectors we verify in this descent have sparsity at most $b$, so the verification cost depends on $b$ rather than on the full input size.

The remaining issue is how to implement this sparse verification in the standard circuit model.  A direct row search would need to read the coordinate $(y_u)_r$ of an internally generated vector at a superposed row index $r$.  Indeed, this is what \cite{AGGSS24} does (see Remark 5.3) but implementing this by a plain selector costs order $n$ gates per row query, which would erase the sparse saving.  Instead, we use block Reed--Solomon (RS) fingerprints. The idea is to divide the rows into blocks of length $\kappa$, scan the candidate vector once to compute an RS fingerprint for each block, and then search over the block indices quantumly.  A block predicate compares the stored RS fingerprint of $y_u$ with the corresponding RS fingerprint of $Mv_u$.  This costs more than the ideal coordinate-access verifier, but it avoids generating complete oracle access to the generated vector.

\subsection{Putting everything together}\label{sec:putting-things-together}

Having shown how to repair a bad node, we are essentially done.  Assume the input $x\in\F^n$, the oracles $U_M$ and $U_\OO$, and a parameter $b\in\N$ are given. We can briefly summarize the algorithm into the following steps:

\paragraph{\textbf{Step 1.}} Construct a binary tree whose root corresponds to $x$ and whose leaves have disjoint supports of size at most $b$.  Initialize the certified list $G$ to be empty.

\paragraph{\textbf{Step 2.}} Sample a BR randomness table $\Omega$.  Given $\Omega$ and $G$, define the badness predicate on the nodes of the tree.  For internal nodes, badness means failure of consistency with the children.  For leaves, badness means failure of the quantum matrix-vector verification test.  Run quantum search over the nodes of the tree to find a bad node.  If none exists, return the output of tournament selection on the root.  If a bad node is found, descend to a minimal bad node, use the repair step above to find a new certified pair $(g,Mg)$, and add it to $G$.

\paragraph{\textbf{Step 3.}} Having found a new certified pair, we have increased the dimension of $W_G$ by one.  Since $V$, and therefore the initial $W_G$, has codimension $O(\log(1/\eps))$, we repeat Step 2 at most $O(\log(1/\eps))$ times.

\paragraph{Runtime:} Let us finish with the informal runtime calculation.  Let $K$ denote the size of the candidate list in an iteration.  Since $|G|=O(\log(1/\eps))$ and the BR sampling parameter is $\ot(1/\eps^{2t+1})$, we have $K=\ot((2/\eps)^{O(\log(1/\eps))})$.  For a fixed node, generating the candidate list and running tournament selection costs roughly $KT$ (recall that $T$ is the runtime of the oracle $\OO$).

The verification cost has a second parameter $\kappa$, the block length used by the RS fingerprinting verifier.  For a leaf of weight at most $b$, the verifier costs, up to logarithmic factors,
\begin{align}\label{eq:sci}
        \Psi(n,b,\kappa)
        :=
        n+
        \left(\frac n\kappa\right)^{3/2}
        +
        b\sqrt{n\kappa}
        \enspace .
\end{align}
The three terms come from scanning the candidate vector once, searching over the RS fingerprint table of length about $n/\kappa$, and computing block RS fingerprints of $Mv_u$.  Thus one bad-node test costs roughly $KT+\Psi(n,b,\kappa)$.

The tree has $O(n/b)$ nodes, so quantum search over the tree costs
\begin{align*}
        \ot\left(\left(KT+\Psi(n,b,\kappa)\right)\sqrt{\frac nb}\right)
        \enspace .
\end{align*}
The descent to a minimal bad node has the same asymptotic cost because the searched subtrees shrink geometrically, and the number of iterations is only $O(\log(1/\eps))$, which is absorbed in $\ot(\cdot)$.  When $T=\Theta(n)$, choosing $b=\kappa=\Theta(n^{1/3})$ gives total running time $\ot(n^{4/3})$, up to the dependence on $\eps$.


\section{Preliminaries}\label{sec:prelim}

Let $\F=\F_q$ be a finite field of size $q$, and for a vector $v\in\F^n$, let
\begin{align*}
        \supp(v)=\{j\in[n]:v_j\ne0\}\enspace ,
        \qquad \wt(v)=|\supp(v)|\enspace .
\end{align*}

We fix the following. All the logarithms are base-2. Let $n\in\N$, and $M\in \F^{n\times n}$ be a matrix. Assume we have an oracle $\OO:\F^n\to\F^n$ such that $\Pr_z[\OO(z)=Mz]\geq \eps$ for some $\eps\in (0,1]$. {In the following sections, we take $\OO$ to be deterministic for simplicity of exposition, and in \cref{apx:randomized-oracles} we give the extension to randomized oracles.} Throughout the paper, $\ot(\cdot)$ hides polylogarithmic factors in $n$, $q$, $1/\eps$, and the relevant inverse error parameters.

\subsection{The Bogolyubov--Ruzsa lemma and self-correction}
In this part, we state the additive-combinatorial idea used by the reduction. Let $A$ be the set of \emph{good} inputs, namely
\begin{align*}
        A:=\{z\in\F^n:\OO(z)=Mz\}\enspace .
\end{align*}
We use the following variation of the probabilistic Bogolyubov--Ruzsa (BR) lemma of Hirahara and Shimizu \cite[Lemma~5.5]{HS26}.

\begin{lemma}[Bogolyubov--Ruzsa]\label{thm:BR}
Let $A\subseteq\F^n$ have density at least $\eps$, where $0<\eps\le1/2$, and set $t\geq\frac{\log(1/\eps)}{2}+1$.
Then there exists an additive subspace $V\subseteq \F^n$ with $\codim(V)\le O(\log(1/\eps))$
such that for every $v\in V$, if
\begin{align*}
        a_1,\ldots,a_t,b_1,\ldots,b_{t-1}\in\F^n
\end{align*}
are sampled independently and uniformly and
\begin{align*}
        b_t:=a_1+\cdots+a_t-b_1-\cdots-b_{t-1}-v\enspace ,
\end{align*}
then
\begin{align*}
        \Prb[a_1,\ldots,a_t,b_1,\ldots,b_t\in A]\ge\eps^{2t+1}\enspace .
\end{align*}
\end{lemma}

\subsubsection*{Self-correction using \cref{thm:BR}}
For $v\in\F^n$, define a randomized BR self-correction sample by drawing $a_1,\ldots,a_t,b_1,\ldots,b_{t-1}$ uniformly, setting
\begin{align*}
        b_t=a_1+\cdots+a_t-b_1-\cdots-b_{t-1}-v\enspace ,
\end{align*}
and outputting
\begin{align*}
        \Bog(v):=\sum_{i=1}^t\OO(a_i)-\sum_{i=1}^t\OO(b_i)\enspace .
\end{align*}
The following simple corollary shows that $\Bog(v)=Mv$ with high probability.
\begin{corollary}\label{cor:bog}
For the subspace $V$ of \cref{thm:BR}, every $v\in V$ satisfies
\begin{align*}
        \Prb[\Bog(v)=Mv]\ge\eps^{2t+1}\enspace .
\end{align*}
\end{corollary}

\begin{proof}
On the event that all sampled $a_1,\ldots,a_t,b_1,\ldots,b_t$ lie in $A$, each oracle value equals the corresponding matrix-vector product. Hence
\begin{align*}
        \Bog(v)
        =\sum_{i=1}^t Ma_i-\sum_{i=1}^t Mb_i
        =M\left(\sum_{i=1}^t a_i-\sum_{i=1}^t b_i\right)=Mv\enspace .
\end{align*}
\cref{thm:BR} lower-bounds this event by $\eps^{2t+1}$.
\end{proof}

\subsection{Tournament selection}\label{sec:tournament}

We will also make use of the following tournament selection procedure from \cite{HS26}. We include a short proof for completeness.

\begin{lemma}[Tournament selection]\label{lem:select}
Given $v\in\F^n$ and an explicit list $\calL=(y^{(1)},\ldots,y^{(K)})$ of vectors in $\F^n$, the procedure $\Select(v,\calL)$ outputs an element of $\calL$ and runs in time $O(K(n+\wt(v)))$.  It uses at most $K\wt(v)$ queries to $U_M$.  If $Mv\in\calL$, then the output is $Mv$.
\end{lemma}

\begin{proof}
Maintain a survivor $y^*$.  Initially set $y^*=y^{(1)}$.  When processing $y^{(k)}$, if $y^{(k)}=y^*$, do nothing.  Otherwise, find a coordinate $r$ on which they differ and compute
\begin{align*}
        (Mv)_r=\sum_{j\in\supp(v)}M_{rj}v_j\enspace .
\end{align*}
If $y^*_r=(Mv)_r$, keep $y^*$; otherwise, replace $y^*$ by $y^{(k)}$.

We prove by induction on the processed prefix that, if $Mv$ occurs among the candidates seen so far, then the survivor equals $Mv$.  The base case is immediate.  If the old survivor is already $Mv$, then it agrees with the true coordinate and is not replaced by an incorrect vector.  If $Mv$ first appears as $y^{(k)}$, then at the distinguishing coordinate, it is the unique one of the two compared vectors that agrees with $(Mv)_r$, so it becomes the survivor.  This proves correctness.  Each comparison scans at most $n$ coordinates and computes one coordinate of $Mv$ in $O(\wt(v))$ time.
\end{proof}

\subsection{Certified directions and candidate lists} \label{sec:Certified directions and candidate lists}
To extend the self-correction idea of \cref{cor:bog} beyond the BR subspace, we maintain a list of input-output pairs of the form
\begin{align*}
        G=\left((g_1,h_1),\ldots,(g_d,h_d)\right)\in(\F^n\times\F^n)^d\enspace ,
\end{align*}
where $h_i$ is intended to be $Mg_i$.  We call $G$ certified if $h_i=Mg_i$ for all $i$.  For $\alpha\in\F^d$, define
\begin{align*}
        w_\alpha:=\sum_{i=1}^d \alpha_i g_i\enspace ,
        \qquad
        h_\alpha:=\sum_{i=1}^d \alpha_i h_i\enspace .
\end{align*}
We will use the space generated by the first components of $G$ together with the hidden BR subspace,
\begin{align*}
        W_G:=V+\spanop\{g_1,\ldots,g_d\}\enspace .
\end{align*}

For a positive integer $R$, a randomness table $\Omega$ for $G$ consists of independent uniform vectors
\begin{align*}
        a_{\alpha,r,1},\ldots,a_{\alpha,r,t}\enspace ,
        b_{\alpha,r,1},\ldots,b_{\alpha,r,t-1}\in\F^n
        \enspace ,
        \qquad \alpha\in\F^d\enspace ,
        \quad r\in[R]
        \enspace .
\end{align*}
For $v\in\F^n$, define the derived vector
\begin{align*}
        b_{\alpha,r,t}(v):=
        \sum_{i=1}^t a_{\alpha,r,i}
        -\sum_{i=1}^{t-1}b_{\alpha,r,i}
        -(v-w_\alpha)
        \enspace .
\end{align*}
The candidate list $\calL_{G,\Omega}(v)$ contains the $q^dR$ vectors
\begin{align*}
        y_{\alpha,r}:=
        \sum_{i=1}^t\OO(a_{\alpha,r,i})
        -\sum_{i=1}^t\OO(b_{\alpha,r,i}(v))
        +h_\alpha
        \enspace ,
\end{align*}
where $b_{\alpha,r,i}(v)=b_{\alpha,r,i}$ for $i<t$.  Thus, the first $2t-1$ random vectors in the table are reused, while the final point is derived from the input vector $v$.

\begin{lemma}\label{lem:list-cost}
Let $K=q^dR$.  Given $G$, $\Omega$, and $v$, the list $\calL_{G,\Omega}(v)$ can be generated explicitly in time $\ot(KT)$.  It uses $2tK=\ot(K)$ queries to $U_{\OO}$ and no queries to $U_M$.
\end{lemma}

\begin{proof}
For each pair $(\alpha,r)$, forming the derived vector $b_{\alpha,r,t}(v)$, making the $2t$ calls to $U_{\OO}$, and adding $h_\alpha$ costs $\ot(T)$ after the table has been sampled.  Updating $w_\alpha$ and $h_\alpha$ across the enumeration of $\alpha$ can be done within the same bound.  Sampling or storing the table costs $\ot(Ktn)$, which is absorbed because $t$ is logarithmic and $T\ge n$.  The only oracle calls during list generation are the $2t$ calls to $U_{\OO}$ made for each of the $K$ candidates.
\end{proof}

\begin{lemma}\label{lem:good-list}
Assume $G$ is certified.  For every fixed $v\in W_G$,
\begin{align*}
        \Pr_\Omega[Mv\in\calL_{G,\Omega}(v)]\ge 1-(1-\eps^{2t+1})^R
        \enspace .
\end{align*}
\end{lemma}

\begin{proof}
Choose $\alpha^*\in\F^d$ such that $w:=v-w_{\alpha^*}\in V$.  For each $r\in[R]$, the tuple indexed by $(\alpha^*,r)$ is a uniform $tA-tA$ representation trial for $w$.  With probability at least $\eps^{2t+1}$, \cref{cor:bog} gives
\begin{align*}
        \sum_{i=1}^t\OO(a_{\alpha^*,r,i})
        -\sum_{i=1}^t\OO(b_{\alpha^*,r,i}(v))
        =Mw
        =M(v-w_{\alpha^*})
        \enspace .
\end{align*}
Because $G$ is certified, $h_{\alpha^*}=Mw_{\alpha^*}$.  The corresponding candidate is therefore $Mv$.  The repetitions over $r$ are independent, so all $R$ repetitions fail with probability at most $(1-\eps^{2t+1})^R$.
\end{proof}

\subsection{Quantum computing} \label{sec:quantum-computing}

We assume familiarity with basics of quantum computation and information and refer the interested reader to standard references such as \cite{NS02}.
We assume the algorithm has quantum oracle access to $\OO$ and the matrix $M$. In particular, it has access to a unitary $U_M$ which gives coherent entry access to the matrix $M$, and $U_\OO$ that gives unitary access to the average-case oracle $\OO$. Formally, we have
\begin{align*}
        U_M:|i,j,\beta\rangle\mapsto |i,j,\beta+M_{ij}\rangle\enspace ,
        \qquad i,j\in[n],\ \beta\in\F\enspace ,
\end{align*}
and
\begin{align*}
        U_{\OO}:|z,y\rangle\mapsto |z,y+\OO(z)\rangle\enspace ,
        \qquad z,y\in\F^n\enspace .
\end{align*}
A query to $U_M$ reveals one matrix entry, and we assume it will cost a constant number of elementary operations.  A query to $U_{\OO}$ writes an $n$-dimensional vector, and we denote the running time of $\OO$ by $T\ge n$ elementary operations.

We work in the standard uniform quantum oracle circuit model.  The only input oracles are $U_M$ and $U_{\OO}$.  Explicit vectors produced by the algorithm, such as selected candidates and sampled tables, are ordinary work registers.  When the circuit needs to use such data at a superposed address, the access is implemented by an explicit reversible circuit and its gate cost is included in the running time.  The matrix-vector product verifier that we will discuss shortly in \cref{lem:restricted-verification} is designed with this point in mind. Instead of reading arbitrary coordinates of a candidate vector inside a row search, it first compresses the candidate into shorter blocks of Reed--Solomon (RS) fingerprints and then searches over blocks. This will allow us to avoid paying $n$ gates when the input vector is sparse. We will discuss this shortly.

All subroutines used inside quantum searches are implemented coherently.  At the implementation level, predicates are written as qubit flags.  A deterministic classical computation is made reversible by computing its work registers, XORing the desired bit into a designated flag qubit, and then uncomputing the work registers.  A bounded-error verification subroutine is used in the same way: measurements are deferred, and the final accept/reject outcome is represented by a coherent flag.  The quantum-search primitives below take such flagging circuits as input.  In particular, we will use fixed-point amplitude amplification in a form that preserves the one-sided property: if the flag has zero amplitude, amplification still leaves it with zero amplitude.  BR randomness tables are sampled before the corresponding quantum search begins and are then held fixed as reversible classical controls throughout that search.  The RS fingerprint seed used in \cref{lem:restricted-verification} is part of the verifier's coherent workspace.

The algorithm is randomized uniform in the usual sense.  For fixed public parameters and a fixed classical random seed, a deterministic classical compiler outputs the corresponding quantum oracle circuit.  The seed determines the BR randomness tables and the classical choices made by the outer search procedures.  The verifier's internal seed registers and amplification circuits are generated uniformly by the circuit itself, which itself will be invoked during the outer search, and hence we require coherency.

\subsubsection*{Quantum search with one-sided predicates}

We will use fixed-point amplitude amplification as our only quantum-search primitive.  There are two forms that we need.  The first one is a measured search: it outputs a marked item or \Fail.  This is used later to find bad nodes in the consistency tree.  The second one is a coherent OR: it only writes a flag saying that a marked item exists.  This is used inside the restricted verifier.

We use the following standard consequence of fixed-point amplitude amplification, in the form given in \cite[Theorem~27]{Gilyen2019}.  

\begin{lemma}[Fixed-point amplitude amplification]\label{lem:fpaa}
Let \(A\) be a unitary, let \(\Pi_{\mathrm{src}}\) be the projection onto a
one-dimensional source space spanned by \(|0\rangle\), and let
\(\Pi_{\mathrm{flag}}\) be a target projector.  Put
\[
        p:=\|\Pi_{\mathrm{flag}}A|0\rangle\|^2
        \enspace .
\]
Given a lower bound \(\lambda>0\) and an error parameter
\(\eta\in(0,1)\), there is a unitary
\[
        \mathsf{Amp}_{\lambda,\eta}(A,\Pi_{\mathrm{flag}})
\]
using
\[
        \ot\left(\lambda^{-1/2}\log(1/\eta)\right)
\]
uses of \(A\), \(A^\dagger\), and the standard QSVT primitives, with the
following guarantees:
\begin{enumerate}
\item If \(p=0\), then the final state has zero support on
\(\Pi_{\mathrm{flag}}\).
\item If \(p\ge\lambda\), then the final state has support on
\(\Pi_{\mathrm{flag}}\) with probability at least \(1-\eta\).
\item More precisely, when \(p>0\), the projection of the final state onto
the flagged subspace is a scalar multiple of $\Pi_{\mathrm{flag}}A|0\rangle$.
\end{enumerate}
\end{lemma}

\begin{proof}
This is the fixed-point amplitude-amplification consequence of quantum
singular value transformation, as in \cite[Theorem~27]{Gilyen2019}.  The
third item is the singular-vector part of the statement: for the projected
unitary
\[
        \Pi_{\mathrm{flag}} A \Pi_{\mathrm{src}} \enspace,
\]
the source space is one-dimensional, so there is only one relevant right
singular vector, namely \(|0\rangle\), and the corresponding left singular
vector is
\[
        \frac{\Pi_{\mathrm{flag}}A|0\rangle}
             {\|\Pi_{\mathrm{flag}}A|0\rangle\|}
        \enspace .
\]
QSVT changes the corresponding singular value according to the fixed-point
amplification polynomial, but it does not change the associated singular
vectors.  Thus the final flagged component is proportional to the original
flagged component.  The first two guarantees are the usual fixed-point
amplification guarantees.
\end{proof}

We now state the search primitive used in the paper.  The coherent part is used inside the restricted verifier.  The measured part is used later to find bad nodes in the consistency tree.

\begin{lemma}\label{lem:search}\label{lem:coherent-or}
Let $\calU$ be a finite set of size $N$, and let $P:\calU\to\{0,1\}$.  Suppose there is a reversible quantum flagging procedure $B$ of cost $C_B$ such that
\begin{align*}
\begin{cases}
        \Pr[B(u)=1]=0, & P(u)=0 \enspace,\\
        \Pr[B(u)=1]\ge 2/3, & P(u)=1
        \enspace .
\end{cases}
\end{align*}
Then, for every $\eta\in(0,1)$, there is a coherent flagging circuit $B_{\mathrm{search}}(\eta)$ with the following guarantees.  If $P(u)=0$ for every $u\in\calU$, then the flag has amplitude zero.  If $P(u)=1$ for at least one $u\in\calU$, then the flag is equal to $1$ with probability at least $1-\eta$.  The running time is
\begin{align*}
        \ot\left(C_B\sqrt N\log(1/\eta)\right)
        \enspace .
\end{align*}

Furthermore, if we measure the flag and the item register at the end, then we get a measured search procedure: it returns a marked element with probability at least $1-\eta$, if one exists, and returns $\Fail$ with probability one if no marked element exists.
\end{lemma}

\begin{proof}
Let \(A\) be the unitary which prepares the uniform superposition over
\(\calU\) and then applies the flagging procedure \(B\):
\[
        A|0\rangle
        =
        \frac{1}{\sqrt N}
        \sum_{u\in\calU}
        |u\rangle B|u\rangle|0\rangle_{\mathrm{work}}
        \enspace .
\]
Let \(\Pi_{\mathrm{flag}}\) project onto states whose designated flag qubit
is \(1\).  If no item is marked, then the one-sided property of \(B\) gives
\[
        \|\Pi_{\mathrm{flag}}A|0\rangle\|^2=0
        \enspace .
\]
If at least one item is marked, then at least one branch has flag probability
at least \(2/3\), and therefore
\[
        \|\Pi_{\mathrm{flag}}A|0\rangle\|^2
        \ge
        \frac{2}{3N}
        \enspace .
\]
Apply \cref{lem:fpaa} with \(\lambda=1/(2N)\) and error parameter \(\eta\).
This gives the coherent OR guarantee and costs
\[
        \ot(C_B\sqrt N\log(1/\eta))
        \enspace .
\]

For the measured version, measure the flag and item registers.  If the flag
is \(0\), return \(\Fail\).  If the flag is \(1\), return the measured item.
It remains to check that this item is marked.  By the one-sided property of
\(B\), the vector $\Pi_{\mathrm{flag}}A|0\rangle$
has support only on branches whose item label \(u\) satisfies \(P(u)=1\).
By the last part of \cref{lem:fpaa}, the final flagged component is a scalar
multiple of this same vector.  Hence the final state has no support on
branches with \(P(u)=0\) and flag \(1\).  Therefore, conditioned on observing
flag \(1\), the measured item is marked.

If no marked item exists, the flag has zero amplitude, so the measured
procedure returns \(\Fail\) with probability one.  If a marked item exists,
the flag is \(1\) with probability at least \(1-\eta\), and conditioned on
this event the measured item is marked.  This proves the measured-search
guarantee.
\end{proof}

\subsection{Restricted verification}\label{sec:fingerprinted-verification}

To do matrix-vector product verification on sparse inputs, we replace the direct row check by a block Reed--Solomon (RS) fingerprinted verifier.  The goal is, given a vector \(v\in\F^n\) supported on a small set and an explicit candidate \(y\in\F^n\), to decide whether \(y=Mv\) with one-sided error.  The difference from the direct row check is that we never assume that a coordinate of \(y\) can be read for free at a superposed row index.  Implementing such an access by a plain selector would cost \(\ot(n)\) gates per row query.

Fix a block length \(\kappa\in[n]\) and partition \([n]\) into \(m=O(n/\kappa)\) consecutive blocks \(B_1,\ldots,B_m\), each of size \(\kappa\), where we assume we pad the last block with 0's for consistency.  Write
\(B_p=\{r^{(p)}_0,\ldots,r^{(p)}_{\kappa-1}\}\).  For an error parameter \(\eta>0\), let \(\Kext/\F\) be an extension field with \(|\Kext|\ge 6\kappa/\eta\). For \(a\in\Kext\) and \(z\in\F^n\), define the RS fingerprint of \(z\) on block \(B_p\) by
\begin{align*}
        \Phi_a(z,p):=\sum_{\ell=0}^{\kappa-1}a^\ell z_{r^{(p)}_\ell}
        \enspace .
\end{align*}

Using this, we state the following lemma, which shows how to use RS fingerprinting to verify matrix-vector product for sparse vectors. In later applications, the verifier is called with an explicit sparse
description of the input vector.  That is, the verifier receives a list of
positions and values $(j,v_j)_{j\in J}$
where \(|J|\) is the sparsity parameter used in the verifier cost.  The cost
of producing this sparse description is not part of
\(\mathsf{Verify}_\kappa\).  It will be charged in the outer circuit that
calls the verifier.

\begin{lemma}[Restricted verification]\label{lem:restricted-verification}
Let \(J\subseteq[n]\), let \(s=|J|\), and let \(v\in\F^n\) be supported on \(J\).  Suppose the values \((v_j)_{j\in J}\) are given explicitly.  For every \(\kappa\in[n]\) and every error parameter \(\eta\in(0,1)\), there is a one-sided quantum flagging verifier
\[
\mathsf{Verify}^M_\kappa(J,(v_j)_{j\in J},y;\eta)
\]
with the following guarantees.  If \(y=Mv\), then the reject flag has amplitude zero.  If \(y\neq Mv\), then the reject flag is equal to \(1\) with probability at least \(1-\eta\).  Its running time is
\begin{align*}
        \ot\left(
        n+
        \left(\frac n\kappa\right)^{3/2}
        +
        s\sqrt{n\kappa}
        \right)
        \enspace ,
\end{align*}
and it uses \(\ot(s\sqrt{n\kappa})\) queries to \(U_M\).  Equivalently, measuring the flag gives an equality test which always accepts when \(y=Mv\) and rejects with probability at least \(1-\eta\) when \(y\neq Mv\).
\end{lemma}

\begin{proof}
We describe the verifier as a unitary flagging circuit.  It has a seed register \(A\), a table register \(\mathrm{tab}\), a reject flag, and workspace registers. The unitary has access to explicit index value pairs $(j,v_j)_{j\in J}$ and the vector $y$ in its input registers. Starting from the clean state, it first prepares the uniform seed superposition
\begin{align*}
        |\psi_0\rangle
        =
        \frac{1}{\sqrt{|\Kext|}}
        \sum_{a\in\Kext}
        |a\rangle_A
        |0\rangle_{\mathrm{tab}}
        |0\rangle_{\mathrm{flag}}
        |0\rangle_{\mathrm{work}}
        \enspace .
\end{align*}
Controlled on the seed \(a\), the verifier scans the explicit vector \(y\) once and computes all block fingerprints
\[
        F_y^a(p):=\Phi_a(y,p)
        \qquad p\in[m].
\]
After this step the state is
\begin{align*}
        |\psi_1\rangle
        =
        \frac{1}{\sqrt{|\Kext|}}
        \sum_{a\in\Kext}
        |a\rangle_A
        |F_y^a\rangle_{\mathrm{tab}}
        |0\rangle_{\mathrm{flag}}
        |0\rangle_{\mathrm{work}}
        \enspace ,
\end{align*}
where \(F_y^a=(F_y^a(1),\ldots,F_y^a(m))\).  The cost of this table computation is \(\ot(n)\) gates and it uses no queries to \(U_M\).

For a fixed seed \(a\), define the exact block predicate
\begin{align*}
        P_{M,J,v,y,a}(p)=1
        \quad\Longleftrightarrow\quad
        F_y^a(p)\neq
        \sum_{\ell=0}^{\kappa-1}a^\ell
        \sum_{j\in J}M_{r^{(p)}_\ell,j}v_j
        \enspace .
\end{align*}
One reversible evaluation of this predicate first reads \(F_y^a(p)\) from the stored table by an ordinary reversible selector.  Since the table has length \(m=O(n/\kappa)\), this costs \(\ot(n/\kappa)\) gates.  It then computes the matrix-side fingerprint
\begin{align*}
        \sum_{\ell=0}^{\kappa-1}a^\ell
        \sum_{j\in J}M_{r^{(p)}_\ell,j}v_j
        \enspace
\end{align*}
using \(O(\kappa s)\) queries to \(U_M\), compares the two fingerprints, writes a block flag, and uncomputes the temporary arithmetic.

Now fix the coherent one-sided OR construction from
\cref{lem:coherent-or} for the domain \([m]\) and error parameter
\(\eta/3\).  We apply this construction coherently, with every
block-predicate evaluation controlled on the seed \(a\) and the stored
table \(F_y^a\).  We first justify that this controlled application is
compatible with the one-dimensional source-space assumption used in
the proof of \cref{lem:coherent-or}.

For each fixed seed \(a\), let \(\mathsf{OR}_a\) denote the coherent
one-sided OR circuit obtained from \cref{lem:coherent-or} for the
predicate \(P_{M,J,v,y,a}\).  If this predicate has a marked block,
then the uniform superposition over \([m]\) has marked probability at
least \(1/m\).  Hence the same lower bound used in the proof of
\cref{lem:coherent-or} applies to every seed \(a\).  In particular, the
fixed-point phase sequence used to construct \(\mathsf{OR}_a\) depends
only on \(m\) and \(\eta\), and not on \(a\).

For a fixed seed \(a\), the source space is spanned by the clean
block-index, flag, and workspace state, and is therefore
one-dimensional.  Thus the fixed-point amplification lemma applies to
\(\mathsf{OR}_a\).  We implement all of the circuits
\(\mathsf{OR}_a\) by a single controlled unitary.  Every call to the
block-predicate circuit and its inverse is controlled on the seed
register \(A\) and the table register \(\mathrm{tab}\), while the
source and flag phase operations act only on the block-index, flag,
and workspace registers.  The resulting unitary leaves \(A\) and
\(\mathrm{tab}\) unchanged and is block diagonal in their
computational basis.  On the branch with seed \(a\) and table
\(F_y^a\), its restriction is exactly \(\mathsf{OR}_a\).  Thus the
fixed-point construction is applied separately inside each
orthogonal seed branch, and not to the span of all seed branches as a
single source space.

For each fixed seed \(a\), write \( |\Gamma_a\rangle \) for the joint
state of the flag and the workspace after \(\mathsf{OR}_a\) is applied
to the predicate \(P_{M,J,v,y,a}\).  The global state is then
\begin{align*}
        |\psi_2\rangle
        =
        \frac{1}{\sqrt{|\Kext|}}
        \sum_{a\in\Kext}
        |a\rangle_A
        |F_y^a\rangle_{\mathrm{tab}}
        |\Gamma_a\rangle_{\mathrm{flag},\mathrm{work}}
        \enspace .
\end{align*}
The verifier's reject flag is the flag written by this controlled OR.
Since the seed states are orthogonal, the total rejection probability
is
\begin{align*}
        \left\|\Pi_{\mathrm{flag}}|\psi_2\rangle\right\|^2
        =
        \frac{1}{|\Kext|}
        \sum_{a\in\Kext}
        \left\|\Pi_{\mathrm{flag}}|\Gamma_a\rangle\right\|^2
        \enspace .
\end{align*}
This identity is the reason the coherent seed behaves like averaging
over a random seed.

If \(y=Mv\), then for every seed \(a\) and every block \(p\), the two fingerprints agree.  Thus \(P_{M,J,v,y,a}(p)=0\) for all \(p\).  By the one-sided property of \cref{lem:coherent-or}, \(\Pi_{\mathrm{flag}}|\Gamma_a\rangle=0\) for every \(a\).  Hence the reject flag has amplitude zero.

Assume now that \(y\neq Mv\), and put \(d:=y-Mv\).  Choose a block \(B_{p^*}\) on which \(d\) is not identically zero.  Then
\begin{align*}
        D_{p^*}(X):=
        \sum_{\ell=0}^{\kappa-1}
        d_{r^{(p^*)}_\ell}X^\ell
        \enspace
\end{align*}
is a nonzero polynomial over \(\Kext\) of degree less than \(\kappa\).  Hence it has at most \(\kappa-1\) roots.  For uniform \(a\in\Kext\),
\begin{align*}
        \Pr_a[D_{p^*}(a)=0]
        \le
        \frac{\kappa}{|\Kext|}
        \le
        \frac{\eta}{6}
        \enspace .
\end{align*}
For every seed with \(D_{p^*}(a)\neq0\), the block \(p^*\) is marked.  Therefore the coherent OR raises the reject flag with probability at least \(1-\eta/3\) on that seed.  Using the averaging identity above,
\begin{align*}
        \Pr[\mathsf{Verify}_\kappa\text{ rejects}]
        &\ge
        \left(1-\frac{\eta}{6}\right)
        \left(1-\frac{\eta}{3}\right)        \\
        &\ge
        1-\eta
        \enspace .
\end{align*}

It remains to account for the cost.  First, note that working on the larger field $\Kext$ incurs additional $\polylog$ operations. The table construction costs \(\ot(n)\).  
The controlled coherent OR uses
\(\ot(\sqrt m\log(1/\eta))
=\ot(\sqrt{n/\kappa}\log(1/\eta))\)
block-predicate evaluations. 
Each block-predicate evaluation costs \(\ot(n/\kappa)\) gates for the selector and \(O(\kappa s)\) queries to \(U_M\).  Thus the running time is
\begin{align*}
        \ot\left(
        n+
        \sqrt{\frac n\kappa}\cdot \frac n\kappa
        +
        \sqrt{\frac n\kappa}\cdot \kappa s
        \right)
        =
        \ot\left(
        n+
        \left(\frac n\kappa\right)^{3/2}
        +
        s\sqrt{n\kappa}
        \right)
        \enspace ,
\end{align*}
and the number of \(U_M\)-queries is \(\ot(s\sqrt{n\kappa})\).
\end{proof}


\section{Main results}
Having defined the necessary tools, we are now ready to state our main result and present a formal proof.

\subsection{Formal statement of the main theorem}

\begin{theorem}[Main result]\label{thm:main}
For every $n\in\mathbb N$, a promise lower bound $\eps$, and $b,\kappa\in[n]$, there is a randomized uniform quantum algorithm $\Red_{b,\kappa}$ such that the following holds.  Let $\F$ be a finite field of size $q$, let $0<\eps\le1/2$ and let $M\in\F^{n\times n}$ and $\OO:\F^n\to\F^n$ satisfy
\begin{align*}
        \Prb_{z\in\F^n}[\OO(z)=Mz]\ge \eps
        \enspace .
\end{align*}
Then, for every $x\in\F^n$,
\begin{align*}
        \Prb[\Red_{b,\kappa}^{\OO,M}(x)=Mx]\ge \frac23
        \enspace .
\end{align*}
The probability is over the internal randomness and quantum measurements. 
If one $U_{\OO}$-query costs $T\ge n$ elementary operations, recalling the definition of $\Psi$ from $\cref{eq:sci}$, then the bounds are
\begin{align*}
        \Params
        {\ot\left((q/\eps)^{O(\log(1/\eps))}\sqrt{\frac nb}\right)}
        {\ot\left((q/\eps)^{O(\log(1/\eps))}\frac{n^{3/2}}{\sqrt b}+n\sqrt{b\kappa}\right)}
        {\ot\left((q/\eps)^{O(\log(1/\eps))}T\sqrt{\frac nb}+\Psi(n,b,\kappa)\sqrt{\frac nb}\right)}
        \enspace .
\end{align*}
Using $T\ge n$, the time bound may equivalently be written as
\begin{align*}
        \ot\left((q/\eps)^{O(\log(1/\eps))}T\sqrt{\frac nb}+\left(\frac n\kappa\right)^{3/2}\sqrt{\frac nb}+n\sqrt{b\kappa}\right)
        \enspace .
\end{align*}

\end{theorem}

\begin{corollary}\label{cor:optimized}

Assume first that $n\le T\le n^{3/2}$, and choose $b=\Theta(T^{4/3}/n)$ and $\kappa=\Theta(n/T^{2/3})$.  Then the bounds are
\begin{align*}
        \Params
        {\ot\left((q/\eps)^{O(\log(1/\eps))}\frac{n}{T^{2/3}}\right)}
        {\ot\left((q/\eps)^{O(\log(1/\eps))}\frac{n^2}{T^{2/3}}+nT^{1/3}\right)}
        {\ot\left((q/\eps)^{O(\log(1/\eps))}nT^{1/3}\right)}
        \enspace .
\end{align*}
In particular, the displayed matrix-query bound is at most $\ot((q/\eps)^{O(\log(1/\eps))}nT^{1/3})$ throughout the range $T\ge n$.  If $T>n^{3/2}$, then choosing $b=n$ and $\kappa=1$ gives
\begin{align*}
        \Params
        {\ot\left((q/\eps)^{O(\log(1/\eps))}\right)}
        {\ot\left((q/\eps)^{O(\log(1/\eps))}n+n^{3/2}\right)}
        {\ot\left((q/\eps)^{O(\log(1/\eps))}T\right)}
        \enspace .
\end{align*}
In particular, when $T=\Theta(n)$, choosing $b=\kappa=\Theta(n^{1/3})$ gives running time and matrix-query complexity $\ot((q/\eps)^{O(\log(1/\eps))}n^{4/3})$.
\end{corollary}

\subsection{The consistency tree}\label{sec:tree}

The starting point in proving \cref{thm:main} is to construct a binary tree in which each node corresponds to a vector.  Fix the requested input $x\in\F^n$ and the block parameter $b\in\{1,\ldots,n\}$ from \cref{thm:main}.  Let $b'=O(n/b)$ be a power of two, by adding zero blocks if necessary.  Build a balanced binary tree $\calT_x$ of height $h$ in which each node is denoted by a string $u\in \bits^{\le h}$.  For a node $u$, we denote its left and right children by $u0$ and $u1$.  Each node $u$ corresponds to a vector, and the root $r=\emptyset$ corresponds to the input $x$.  The vectors corresponding to the leaves are obtained by partitioning $[n]$ into blocks of size at most $b$ and letting $v_1,\ldots,v_{b'}$ be the corresponding restrictions of $x$ to each block, together with the added zero leaves.  Thus $\wt(v_i)\le b$.  For every  node $u$, let $v_u$ be the sum of the leaf vectors in the subtree rooted at $u$, so the root $r$ satisfies $v_r=x$.  Also, denote the number of nodes in the tree by $N_{\calT}$, so that $N_{\calT}=O(n/b)$.

Having defined the tree, we now move on to defining what we denote as a bad node.  First, for a fixed set $G$ and fixed randomness table $\Omega$, define the selected value at a node by $y_u:=\Select(v_u,\calL_{G,\Omega}(v_u))$.  Note that the selected values are not assumed to be correct, since $v_u$ may not necessarily be in $W_G$.  However, our focus in defining bad nodes is not necessarily correctness, but consistency with the children.

\begin{definition}[Bad node]\label{def:bad}
A leaf $u$ is \emph{bad} if $y_u\ne Mv_u$.  An internal node $u$ with children $u0,u1$ is \emph{bad} if $y_u\ne y_{u0}+y_{u1}$.
A bad node is \emph{minimal} if none of its proper descendants is bad.
\end{definition}

The next observation shows that if the tree has no bad nodes, then we are done.

\begin{observation}\label{lem:no-bad}
If the subtree rooted at $u$ contains no bad node, then $y_u=Mv_u$.  In particular, if the whole tree has no bad node, then $y_r=Mx$.
\end{observation}

\begin{proof}
We prove this by induction on the height of the subtree.  For a leaf, the assertion is exactly the definition of not being bad.  For an internal node, both child subtrees contain no bad node, so the induction hypothesis gives $y_{u0}=Mv_{u0}$ and $y_{u1}=Mv_{u1}$.  Since $u$ is not bad,
\begin{align*}
        y_u=y_{u0}+y_{u1}=Mv_{u0}+Mv_{u1}=Mv_u\enspace . 
\end{align*}
\end{proof}

Next, we show how to make progress by finding a minimal bad node.  In this part, we assume there exists a procedure called $\Extract$ such that, when called on a leaf $u$ and $G$ with $v_u\notin W_G$, it outputs a pair $(z,Mz)$ where $z\notin W_G$.

\begin{lemma}\label{lem:minimal}
Assume $G$ is certified, and assume $\Omega$ is a \emph{good} randomness table, in the following sense (we often call this the good-list event).  For every $u\in\Nodes(\calT_x)$, we have
\begin{align*}
        v_u\in W_G \quad\Longrightarrow\quad Mv_u\in\calL_{G,\Omega}(v_u)\enspace .
\end{align*}
Under these conditions, let $u$ be a minimal bad node.  Then:
\begin{enumerate}[label=(\alph*)]
\item If $u$ is internal with children $u0,u1$, then $v_u\notin W_G$ and $y_{u0}+y_{u1}=Mv_u$.  Thus, appending $(v_u,y_{u0}+y_{u1})$ to $G$ keeps $G$ certified and increases the quotient dimension by at least one.
\item If $u$ is a leaf, then $v_u\notin W_G$.  Hence $\Extract(v_u,G)$ satisfies the precondition and outputs $(z,Mz)$.  Adding this pair to $G$ increases the quotient dimension by one.
\end{enumerate}
\end{lemma}

\begin{proof}
Suppose first that $u$ is internal.  By minimality, neither child subtree contains a bad node.  \cref{lem:no-bad} gives $y_{u0}=Mv_{u0}$ and $y_{u1}=Mv_{u1}$, hence $y_{u0}+y_{u1}=Mv_u$.  Since $u$ is bad, $y_u\ne Mv_u$, but if $v_u\in W_G$, the good-table assumption and \cref{lem:select} would imply $y_u=Mv_u$.  Therefore $v_u\notin W_G$.  Hence, adding $(v_u,y_{u0}+y_{u1})$ to $G$ increases the dimension of $V+\spanop\{g:(g,h)\in G\}$ by at least one.

If $u$ is a leaf, badness means $y_u\ne Mv_u$.  If $v_u\in W_G$, the good-list event and selection would again imply $y_u=Mv_u$, a contradiction.  Thus $v_u\notin W_G$.  Applying $\Extract$ to the leaf produces a certified new direction on the good event.
\end{proof}

The previous lemma shows that in order to make progress towards increasing the space of elements we can self-correct, it suffices to find a minimal bad node.  This is exactly what we do using a quantum search.  To do so and apply \cref{lem:search}, we need to show that we can implement this predicate with one-sided error.  This is what we show next.

\begin{lemma}\label{lem:bad-cost}
Fix $G$ and $\Omega$, let $K=q^{|G|}R$, and define the ideal local predicate
\begin{align*}
        P_{\mathrm{bad}}^{G,\Omega}(u)=1
        \quad\Longleftrightarrow\quad
        u\text{ is bad in the sense of \cref{def:bad}}
        \enspace .
\end{align*}
For every $\eta\in(0,1)$, there is a coherent one-sided flagging circuit $B_{\mathrm{bad}}^{G,\Omega}(\eta)$ for this predicate.  On a clean input $|u\rangle|0\rangle$, it has no false positives:
\begin{align*}
        P_{\mathrm{bad}}^{G,\Omega}(u)=0
        \quad\Longrightarrow\quad
        \left\|
        \Pi_{\mathrm{flag}}
        B_{\mathrm{bad}}^{G,\Omega}(\eta)|u\rangle|0\rangle\right\|^2=0
        \enspace ,
\end{align*}
and it raises the flag with high probability on bad nodes:
\begin{align*}
        P_{\mathrm{bad}}^{G,\Omega}(u)=1
        \quad\Longrightarrow\quad
        \left\|
        \Pi_{\mathrm{flag}}
        B_{\mathrm{bad}}^{G,\Omega}(\eta)|u\rangle|0\rangle
        \right\|^2
        \ge 1-\eta
        \enspace .
\end{align*}
Equivalently, measuring the flag gives a one-sided predicate which never outputs $1$ on a non-bad node and outputs $1$ with probability at least $1-\eta$ on a bad node.

The resource bounds for one call to the flagging circuit are
\begin{align*}
        \Params
        {\ot\left(K\log(1/\eta)\right)}
        {\ot\left((Kn+b\sqrt{n\kappa})\log(1/\eta)\right)}
        {\ot\left((KT+\Psi(n,b,\kappa))\log(1/\eta)\right)}
        \enspace .
\end{align*}
\end{lemma}

\begin{proof}
All randomness in \(\Omega\) is sampled before the quantum search begins and
is then held fixed as a reversible classical control.  Thus
\(P_{\mathrm{bad}}^{G,\Omega}\) is a fixed predicate during the search.

We describe the flagging circuit on an arbitrary superposition of nodes.  Let
\[
        |\psi_0\rangle
        =
        \sum_{u\in\Nodes(\calT_x)}
        \alpha_u
        |u\rangle
        |0\rangle_{\mathrm{type}}
        |0\rangle_{\mathrm{data}}
        |0\rangle_{\mathrm{flag}}
        |0\rangle_{\mathrm{work}}
        \enspace .
\]
The register \(\mathrm{type}\) will store 0 for an
internal node, and 1 if it is a leaf.  The register \(\mathrm{data}\) stores the selected values
and, for leaves, the explicit sparse description that is passed to the
restricted verifier.

First, the circuit computes the type of \(u\), and also the relevant
intervals in the tree.  This uses only reversible arithmetic on the node
label.  For every node \(u\), it can compute the interval of leaves below
\(u\), and hence the support interval of the vector \(v_u\).  This costs only
polylogarithmic factors.

We next compute the selected values.  Given \(u\), the circuit materializes
the vector \(v_u\) from the input \(x\), generates the list
\(\calL_{G,\Omega}(v_u)\), and runs the deterministic tournament procedure
\[
        y_u:=\Select(v_u,\calL_{G,\Omega}(v_u))
        \enspace .
\]
This whole computation is reversible by keeping the work registers for list
generation and tournament selection.  By \cref{lem:list-cost,lem:select},
one selected value costs \(\ot(KT)\) time, uses \(\ot(K)\) queries to
\(U_\OO\), and uses at most \(\ot(Kn)\) queries to \(U_M\).  The cost of materializing \(v_u\) by scanning \(x\) is
\(\ot(n)\), and is absorbed because \(T\ge n\).

After this step, the state has the form
\[
\begin{aligned}
        |\psi_1\rangle
        =
        &\sum_{u\ \mathrm{internal}}
        \alpha_u
        |u\rangle
        |\mathrm{0}\rangle
        |y_u,y_{u0},y_{u1}\rangle
        |0\rangle_{\mathrm{flag}}
        |\sigma_u\rangle                                      \\
        &+
        \sum_{u\ \mathrm{leaf}}
        \alpha_u
        |u\rangle
        |\mathrm{1}\rangle
        |y_u\rangle
        |0\rangle_{\mathrm{flag}}
        |\sigma_u\rangle
        \enspace .
\end{aligned}
\]
Here \(|\sigma_u\rangle\) denotes the remaining workspace.  For an
internal node we compute three selected values, namely those for \(u,u0,u1\).
For a leaf we compute only \(y_u\) at this stage.

Now consider the internal-node branches.  On such a branch, the circuit
computes the exact bit
\[
        f_{\mathrm{int}}(u)
        :=
        \mathbf 1[y_u\ne y_{u0}+y_{u1}]
        \enspace ,
\]
XORs this bit into the designated flag qubit, and uncomputes the temporary
comparison workspace.  Thus on internal nodes the flag is exactly the
badness predicate:
\[
        f_{\mathrm{int}}(u)
        =
        P_{\mathrm{bad}}^{G,\Omega}(u)
        \enspace .
\]

It remains to describe the leaf branches.  For a leaf \(u\), let \(I_u\) be
the interval of coordinates corresponding to that leaf.  The circuit now
prepares the explicit sparse description of \(v_u\) that will be given to
the restricted verifier.  It scans the input register \(x\) once and writes
\[
        \bigl(j,x_j\bigr)_{j\in I_u}
\]
into a workspace register.  Since the leaves are consecutive blocks,
this is done coherently as follows.  For each coordinate \(r\in[n]\), write
\(r=qb+\ell\).  On the branch indexed by \(u\), the circuit checks whether
\(q\) is the block index of the vector corresponding to leaf \(u\).  If so, it copies \(x_r\) into
the fixed output slot \(\ell\).  This is a reversible scan over \(x\), costs
\(\ot(n)\) gates, and is performed only once before the verifier is called.
The resulting state on the leaf part is
\[
        \sum_{u\ \mathrm{leaf}}
        \alpha_u
        |u\rangle
        |\mathrm{1}\rangle
        |y_u\rangle
        |(j,x_j)_{j\in I_u}\rangle
        |0\rangle_{\mathrm{flag}}
        |\sigma'_u\rangle
        \enspace .
\]
This is the point at which the outer circuit has produced the explicit
values needed by the inner verifier.

The circuit now invokes the restricted verifier
\[
 \mathsf{Verify}^M_\kappa
        \bigl(I_u,(x_j)_{j\in I_u},y_u;\eta\bigr)
\]
on the leaf branches. Since
\(|I_u|\le b\), \cref{lem:restricted-verification} gives running time
\[
        \ot\bigl(\Psi(n,b,\kappa)\log(1/\eta)\bigr)
        \enspace
\]
and \(\ot(b\sqrt{n\kappa}\log(1/\eta))\) queries to \(U_M\).  If
\(y_u=Mv_u\), the verifier's reject flag has amplitude zero.  If
\(y_u\ne Mv_u\), the reject flag is equal to \(1\) with probability at least
\(1-\eta\).

Thus, after the leaf verifier is applied, the state has the form
\[
\begin{aligned}
        |\psi_2\rangle
        =
        &\sum_{u\ \mathrm{internal}}
        \alpha_u
        |u\rangle
        |\mathrm{0}\rangle
        |y_u,y_{u0},y_{u1}\rangle
        |f_{\mathrm{int}}(u)\rangle_{\mathrm{flag}}
        |\tau_u\rangle                                      \\
        &+
        \sum_{u\ \mathrm{leaf}}
        \alpha_u
        |u\rangle
        |\mathrm{1}\rangle
        |y_u\rangle
        |(j,x_j)_{j\in I_u}\rangle
        |\Gamma_u\rangle_{\mathrm{flag},\mathrm{work}}
        \enspace .
\end{aligned}
\]
For a leaf \(u\), the verifier guarantee says
\[
        y_u=Mv_u
        \quad\Longrightarrow\quad
        \|\Pi_{\mathrm{flag}}|\Gamma_u\rangle\|^2=0
        \enspace ,
\]
and
\[
        y_u\ne Mv_u
        \quad\Longrightarrow\quad
        \|\Pi_{\mathrm{flag}}|\Gamma_u\rangle\|^2\ge 1-\eta
        \enspace .
\]

We now verify the one-sided predicate guarantee.  If \(u\) is an internal
node, the flag equals \(P_{\mathrm{bad}}^{G,\Omega}(u)\) exactly.  If \(u\)
is a leaf, then \(u\) is bad exactly when \(y_u\ne Mv_u\), and the restricted
verifier has the one-sided guarantee above.  Therefore, for every basis
state \(u\),
\[
        P_{\mathrm{bad}}^{G,\Omega}(u)=0
        \quad\Longrightarrow\quad
        \left\|
        \Pi_{\mathrm{flag}}
        B_{\mathrm{bad}}^{G,\Omega}(\eta)|u\rangle|0\rangle \right\|^2=0
        \enspace ,
\]
and
\[
        P_{\mathrm{bad}}^{G,\Omega}(u)=1
        \quad\Longrightarrow\quad
        \left\|
        \Pi_{\mathrm{flag}}
        B_{\mathrm{bad}}^{G,\Omega}(\eta)|u\rangle|0\rangle
        \right\|^2
        \ge 1-\eta
        \enspace .
\]
The construction is controlled by the classical node label \(u\), so it is
block diagonal in the node register.  Hence the same guarantee holds on
arbitrary superpositions of nodes.  In particular, measuring the flag gives
a one-sided implementation of the bad-node predicate.

It remains to account for the resources.  On an internal node, the circuit
computes three selected values and an exact equality check.  This costs
\(\ot(KT)\) time, \(\ot(K)\) queries to \(U_\OO\), and \(\ot(Kn)\) queries
to \(U_M\), up to constant factors.  On a leaf, the circuit computes one
selected value, scans \(x\) once to write the leaf slice, and runs the
restricted verifier.  The scan of \(x\) costs \(\ot(n)\) gates and is
absorbed in the \(n\) term of \(\Psi(n,b,\kappa)\).  The verifier contributes
\(\ot(\Psi(n,b,\kappa)\log(1/\eta))\) time and
\(\ot(b\sqrt{n\kappa}\log(1/\eta))\) queries to \(U_M\).  Using a common
\(\log(1/\eta)\) factor for the statement, the resource bounds are
\[
        \Params
        {\ot\left(K\log(1/\eta)\right)}
        {\ot\left((Kn+b\sqrt{n\kappa})\log(1/\eta)\right)}
        {\ot\left((KT+\Psi(n,b,\kappa))\log(1/\eta)\right)}
        \enspace .
\]
This proves the lemma.
\end{proof}

Having defined the predicate, we now move on to show how to use it within a quantum search to find a minimal bad node.

\begin{lemma}\label{lem:find-minimal}
Fix $G$ and $\Omega$, and an error parameter $\delta\in(0,1)$.  If a bad node exists, then a minimal bad node can be found with probability at least $1-\delta$.  If no bad node exists, the first search returns \Fail with probability one.  The resource bounds are
\begin{align*}
        \Params
        {\ot\left(K\sqrt{N_\calT}\log(1/\delta)\right)}
        {\ot\left((Kn+b\sqrt{n\kappa})\sqrt{N_\calT}\log(1/\delta)\right)}
        {\ot\left((KT+\Psi(n,b,\kappa))\sqrt{N_\calT}\log(1/\delta)\right)}
        \enspace .
\end{align*}
\end{lemma}

\begin{proof}
We first define a search subroutine $\mathsf{Search}(S,\rho)$ for a set $S$ of nodes and an error parameter $\rho\in(0,1)$.  It applies \cref{lem:search} to the set $S$, using the bad-node flagging circuit $B_{\mathrm{bad}}^{G,\Omega}(1/10)$ from \cref{lem:bad-cost} as the one-sided flagging procedure.  This circuit never flags a non-bad node, and it flags a bad node with probability at least $9/10$, so it satisfies the hypothesis of \cref{lem:search} after adjusting constants.

Therefore, if $S$ contains a bad node, then $\mathsf{Search}(S,\rho)$ returns a bad node with probability at least $1-\rho$.  If $S$ contains no bad node, then $\mathsf{Search}(S,\rho)$ returns \Fail with probability one.  The resource bounds for $\mathsf{Search}(S,\rho)$ are
\begin{align*}
        \Params
        {\ot\left(K\sqrt{|S|}\log(1/\rho)\right)}
        {\ot\left((Kn+b\sqrt{n\kappa})\sqrt{|S|}\log(1/\rho)\right)}
        {\ot\left((KT+\Psi(n,b,\kappa))\sqrt{|S|}\log(1/\rho)\right)}
        \enspace .
\end{align*}

First run $\mathsf{Search}(\Nodes(\calT_x),\rho_0)$, where $\rho_0$ will be chosen below.  If the tree has no bad node, then the searched set contains no bad node, and by the one-sided property above the search returns \Fail with probability one.  This proves the no-bad-node part of the lemma.

Assume now that at least one bad node exists.  Except with probability $\rho_0$, the first search returns some bad node $u$.  Starting from this node, we descend as follows.  If $u$ is a leaf, output $u$.  If $u$ is internal, let $S_0$ and $S_1$ be the node sets of the two child subtrees rooted at $u0$ and $u1$.  First run $\mathsf{Search}(S_0,\rho_i)$.  If this returns a bad node, move to that node and continue.  If it returns \Fail, run $\mathsf{Search}(S_1,\rho_i)$.  If this returns a bad node, move to that node and continue.  If both child searches return \Fail, output the current node $u$.

The one-sided no-false-positive property is crucial.  A search over a subtree containing no bad node returns \Fail with probability one, so the descent never moves into a clean subtree.  Conversely, if a child subtree contains a bad node, the only way to miss all bad nodes in that subtree is the allotted failure event of that search.  Therefore, outside the union of all search-failure events, every move is to a subtree containing a bad node, and the algorithm stops exactly when neither child subtree contains a bad node.  The output is then bad and has no bad proper descendant, namely it is a minimal bad node.

Choose, for example, $\rho_0\le \delta/100$ and $\rho_i=\delta/(100(i+1)^2)$ for the searches made at depth $i$.  There are at most $O(\log N_\calT)$ depths and at most two child searches per depth, so the sum of the failure probabilities is at most $\delta$, after changing constants.

It remains to bound the cost.  Along the descent, the searched subtree sizes decrease geometrically.  If the current subtree has size $s$, each child subtree has size at most $s/2$ up to harmless rounding.  Thus, even charging for both child searches at every level, the sum of square roots of searched domain sizes is bounded by
\begin{align*}
        \sqrt{N_\calT}
        +2\sqrt{N_\calT/2}
        +2\sqrt{N_\calT/4}
        +\cdots
        =
        O(\sqrt{N_\calT})
        \enspace .
\end{align*}
The confidence parameters contribute only logarithmic factors in $N_\calT$ and $1/\delta$, which are hidden in $\ot(\cdot)$ except for the displayed $\log(1/\delta)$ dependence.  Multiplying by the bad-node flagging cost from \cref{lem:bad-cost} gives the stated resource bounds.
\end{proof}

\subsection{Sparse extraction}\label{sec:extraction}

In the statement of \cref{lem:minimal}, we assumed the existence of a subroutine $\Extract$ such that, given a bad leaf, it finds a quotient direction to increase the dimension of $W_G$.  Here, we state how this procedure works.  The subroutine is essentially a classical binary search below the $b$-sparse leaf.

\begin{tcolorbox}[title=Recursive procedure: \textsc{Extract}]
\paragraph{Input:} Oracles $\OO$ and $M$, current set $G$, and a $b$-sparse vector $v\in\F^n$ which is known to be bad.

\paragraph{Goal:} Add one new certified pair to $G$.

\begin{itemize}
    \item If $\operatorname{wt}(v)=1$, compute $Mv$ directly using $M$, update $G\gets G\cup\{(v,Mv)\}$, and return.

    \item Otherwise, split $v$ into $v=v_0+v_1$, where each vector has half the non-zero support of $v$, by extending the tree one layer further.

    \item Sample fresh BR randomness tables $\Omega_0$ and $\Omega_1$ for the two children, and compute
    \begin{align*}
        y_0=\Select\bigl(v_0,\calL_{G,\Omega_0}(v_0)\bigr)\enspace ,
        \qquad
        y_1=\Select\bigl(v_1,\calL_{G,\Omega_1}(v_1)\bigr)
        \enspace .
    \end{align*}

    \item Run restricted verification for the two claims $y_0=Mv_0$ and $y_1=Mv_1$.

    \begin{itemize}
        \item If the first verification rejects, call $\textsc{Extract}(\OO,M,G,v_0)$.

        \item Otherwise, if the second verification rejects, call $\textsc{Extract}(\OO,M,G,v_1)$.

        \item Otherwise, both children have verified successfully.  Update $G\gets G\cup\{(v,y_0+y_1)\}$.
    \end{itemize}
\end{itemize}
\end{tcolorbox}

The intended logic is simple.  A bad child cannot lie in $W_G$ on the good-list event, because the good-list event and deterministic selection would make its selected value correct.  If neither child is detected as bad, then both selected child values are correct and their sum gives $Mv$.

\begin{lemma}[Extraction progress]\label{lem:extract}
Assume that the initial list $G$ is certified.  Suppose that throughout one invocation of $\Extract(v,G)$ the following two conditions hold for the current list $G$:
\begin{enumerate}
\item Every generated child list $\calL_i$ for a child vector $v_i\in W_G$ contains $Mv_i$;
\item Every extraction equality test whose selected value is not the correct product returns $\No$.
\end{enumerate}
If $v\notin W_G$ and $\wt(v)\le b$, then $\Extract(v,G)$ returns after at most $O(\log b)$ recursive levels, appends one certified pair $(g,Mg)$, and increases
\begin{align*}
        \dim\bigl(V+\spanop\{g':(g',h')\in G\}\bigr)
        \enspace
\end{align*}
by at least one.
\end{lemma}

\begin{proof}
On every recursive call, the support size is at most half the previous support size, so the recursion depth is at most $O(\log b)$.

We prove the invariant that the current vector $v'$ on the recursive path satisfies $v'\notin W_G$, where $G$ is still the list as it was at the start of this extraction call.  The invariant holds for the initial vector by assumption.  If $v'=v_0+v_1$ and the procedure recurses on a child $v_i$ detected as bad, then $y_i\ne Mv_i$: restricted verification has no false positives, so it cannot return $\No$ when $y_i=Mv_i$.  If $v_i\in W_G$, condition (1) and \cref{lem:select} would force $y_i=Mv_i$, a contradiction.  Thus $v_i\notin W_G$, and the invariant is preserved along the recursive path.

If the procedure reaches the weight-one branch, it reads the exact column combination $Mv'$ and appends $(v',Mv')$.  Since $v'\notin W_G$, the dimension increases by at least one.

It remains to consider a terminal branch in which neither child of the current vector $v'=v_0+v_1$ is detected as bad.  By condition (2), the selected child values $y_0$ and $y_1$ must be correct.  Hence, the appended pair is
\begin{align*}
        (v',y_0+y_1)=(v',Mv_0+Mv_1)=(v',Mv')\enspace .
\end{align*}
Again $v'\notin W_G$, so this certified append increases the quotient dimension by at least one.
\end{proof}

\subsection{Full algorithm, correctness, and bounds}
\label{sec:full-algorithm}

We now put the pieces together.  The parameter $t$ is the one from \cref{thm:BR}, so $t=O(\log(1/\eps))$.  We take the number of BR samples in each candidate list to be
\begin{align*}
        R=\ot\left(\frac{1}{\eps^{2t+1}}\right)
        \enspace .
\end{align*}
More concretely, $R$ is chosen large enough so that a union bound over all iterations, all tree nodes, and all vectors queried during extraction succeeds with constant probability.  This only costs logarithmic factors, which are hidden in $\ot(\cdot)$.  We will run $O(\log(1/\eps))$ iterations, since the BR subspace has codimension $O(\log(1/\eps))$.

\begin{tcolorbox}[title=Final algorithm]
\paragraph{Input:} Oracles $\OO$ and $M$, input vector $x\in\F^n$, and parameters $\eps, b,\kappa$.

\paragraph{Output:} A vector in $\F^n$.

\begin{itemize}
    \item Build the consistency tree $\calT_x$ from \cref{sec:tree}.  The root vector is $x$, each leaf has weight at most $b$, and the number of nodes is $N_{\calT}=O(n/b)$.

    \item Initialize $G\gets\emptyset$.

    \item Repeat the following for $O(\log(1/\eps))$ iterations.

    \begin{itemize}
        \item Sample a fresh BR randomness table $\Omega$ for the current set $G$.

        \item For a node $u$, define
        \begin{align*}
                y_u=\Select\bigl(v_u,\calL_{G,\Omega}(v_u)\bigr)
                \enspace .
        \end{align*}

        \item Use quantum search over the nodes of $\calT_x$ to find a bad node in the sense of \cref{def:bad}.  If a bad node is found, descend to a minimal bad node.

        \item If no bad node is found, output $\Select\bigl(x,\calL_{G,\Omega}(x)\bigr)$ and return.

        \item If the minimal bad node $u$ is internal, append $(v_u,y_{u0}+y_{u1})$ to $G$.

        \item If the minimal bad node $u$ is a leaf, call $\Extract(v_u,G)$.
    \end{itemize}
\end{itemize}
\end{tcolorbox}

\subsubsection*{Correctness}

We prove that the algorithm outputs \(Mx\) with probability at least \(2/3\).
Let $L:=\left\lceil C_0\log(1/\eps)\right\rceil$ be the number of iterations, where \(C_0\) is chosen sufficiently large that $L>\codim(V)$
for the subspace \(V\) supplied by \cref{thm:BR}.  Also, let
\[
        H:=\left\lceil\log_2 b\right\rceil+1
\]
be an upper bound on the number of recursive levels in any call to
\(\Extract\).  Recall that the consistency tree has
\(N_{\calT}=O(n/b)\) nodes.

We choose
\begin{align*}
        R
        =
        \left\lceil
        \frac{C_1}{\eps^{2t+1}}
        \log\bigl(30L(N_{\calT}+2H)\bigr)
        \right\rceil
\end{align*}
for a sufficiently large constant \(C_1\).  In particular,
\[
        R=\widetilde O(1/\eps^{2t+1})\enspace.
\]

Put
\[
        p_{\mathrm{list}}
        :=
        (1-\eps^{2t+1})^R
        \le
        \exp(-\eps^{2t+1}R)
        \le
        \frac{1}{30L(N_{\calT}+2H)}\enspace.
\]

We define three bad events.  The first event is guarded by the condition
that the current list \(G\) is certified. This is important because
\cref{lem:good-list} assumes certification.

\begin{itemize}
    \item Let \(\calE_{\mathrm{list}}\) be the event that, at some point
    during the execution, an instance
    \(\calL_{G,\Omega}(v)\) is generated such that $G$ is certified, $v\in W_G$, and $Mv\notin\calL_{G,\Omega}(v)$.

    We claim that
    \[
            \Pr[\calE_{\mathrm{list}}]\le \frac{1}{30} \enspace.
    \]
    To see this, first consider the tree-node lists in an outer iteration.
    Condition on the complete execution transcript immediately before the
    fresh outer randomness table \(\Omega\) is sampled.  At this point the
    current list \(G\) is fixed.  If \(G\) is certified, then for every
    fixed tree node \(u\) satisfying \(v_u\in W_G\),
    \cref{lem:good-list} gives
    \[
        \Pr_{\Omega}\left[
            Mv_u\notin\calL_{G,\Omega}(v_u)
            \,\middle|\,
            \text{previous transcript}
        \right]
        \le p_{\mathrm{list}} \enspace.
    \]
    If \(G\) is not certified, the corresponding execution does not
    contribute to \(\calE_{\mathrm{list}}\).  Thus, by a union bound over
    the \(N_{\calT}\) nodes, the probability of a guarded list failure in
    this outer iteration is at most
    \(N_{\calT}p_{\mathrm{list}}\).  Notice that independence between the
    node events is not needed, even though all node lists use the same
    table \(\Omega\).

    Next consider lists generated inside \(\Extract\).  Immediately before
    each child randomness table is sampled, condition on the entire
    previous transcript.  The current vector \(v\) and the current list
    \(G\) are then fixed, and the child table is fresh.  Therefore, whenever
    \(G\) is certified and \(v\in W_G\), the same conditional bound
    \(p_{\mathrm{list}}\) follows from \cref{lem:good-list}.  There are at
    most \(2H\) child-list instances in one call to \(\Extract\).

    There are at most \(L\) outer iterations and at most one extraction call
    per iteration.  Consequently,
    \[
        \Pr[\calE_{\mathrm{list}}]
        \le
        L(N_{\calT}+2H)p_{\mathrm{list}}
        \le
        \frac{1}{30} \enspace.
    \]

    \item Let \(\calE_{\mathrm{ver}}\) be the event that some measured
    verification call inside \(\Extract\) accepts an incorrect equality.
    Set the error parameter of each such verifier to
    \[
            \eta_{\mathrm{ver}}
            :=
            \frac{1}{20LH} \enspace.
    \]
    There are at most \(2LH\) measured verification calls during the whole
    execution.  For any fixed, possibly adaptively chosen, input
    \((v,y)\), the verifier always accepts when \(y=Mv\), and accepts with
    probability at most \(\eta_{\mathrm{ver}}\) when \(y\ne Mv\).
    Conditioning on the transcript before each invocation and applying a
    union bound therefore gives
    \[
            \Pr[\calE_{\mathrm{ver}}]
            \le
            2LH\eta_{\mathrm{ver}}
            \le
            \frac{1}{10} \enspace.
    \]

    \item Let \(\calE_{\mathrm{search}}\) be the event that some invocation
    of \cref{lem:find-minimal} fails to satisfy its guarantee. Namely, the
    tree contains a bad node, but the invocation does not return a minimal
    bad node.  We run each invocation with failure probability at most
    \(1/(10L)\).  Conditioned on the current \(G\) and \(\Omega\), the
    guarantee of \cref{lem:find-minimal} applies to the resulting fixed
    bad-node predicate.  Since there are at most \(L\) invocations,
    \[
            \Pr[\calE_{\mathrm{search}}]
            \le
            \frac{1}{10} \enspace.
    \]
    If the tree contains no bad node, the search returns \(\Fail\) with
    probability one by the one-sided guarantee, so this case contributes
    no failure probability.
\end{itemize}
Let
\[
        \calE
        :=
        \neg\calE_{\mathrm{list}}
        \cap
        \neg\calE_{\mathrm{ver}}
        \cap
        \neg\calE_{\mathrm{search}} \enspace.
\]
By the union bound,
\[
        \Pr[\calE]
        \ge
        1-\frac{1}{30}-\frac{1}{10}-\frac{1}{10}
        >
        \frac{2}{3} \enspace.
\]
It remains to prove that the algorithm outputs \(Mx\) whenever \(\calE\)
occurs.

We prove by induction over the iterations reached by the algorithm that
the current list \(G\) is certified and that every non-terminating
iteration appends a certified pair whose first component lies outside the
current space \(W_G\).  Initially \(G=\emptyset\), so the certification
invariant holds.

Assume that \(G\) is certified at the beginning of an iteration, and write
\[
        W_G
        :=
        V+\spanop\{g:(g,h)\in G\} \enspace.
\]
Because \(\calE_{\mathrm{list}}\) does not occur and \(G\) is certified,
the freshly sampled outer table has the following property for every tree
node \(u\):
\[
        v_u\in W_G
        \quad\Longrightarrow\quad
        Mv_u\in\calL_{G,\Omega}(v_u) \enspace.
\]
Hence, by \cref{lem:select},
\[
        v_u\in W_G
        \quad\Longrightarrow\quad
        y_u=Mv_u \enspace.
\]

If the search returns \(\Fail\), then, because
\(\calE_{\mathrm{search}}\) does not occur and the no-bad-node case is
one-sided, the tree contains no bad node.  By \cref{lem:no-bad}, the
selected value at the root is correct.  Since the root vector is \(x\),
the algorithm outputs
\[
        \Select(x,\calL_{G,\Omega}(x))=Mx \enspace,
\]
and terminates correctly.

Otherwise, the search returns a minimal bad node \(u\).  Suppose first
that \(u\) is internal, with children \(u0\) and \(u1\).  Since \(u\) is
minimal, neither child subtree contains a bad node.  By
\cref{lem:no-bad}, we have $y_{u0}=Mv_{u0}$and $ y_{u1}=Mv_{u1}$.
Therefore,
\[
        y_{u0}+y_{u1}
        =
        Mv_{u0}+Mv_{u1}
        =
        Mv_u \enspace.
\]
Since \(u\) is bad,
\[
        y_u\ne y_{u0}+y_{u1}=Mv_u \enspace.
\]
If \(v_u\in W_G\), then the absence of
\(\calE_{\mathrm{list}}\) and \cref{lem:select} would imply
\(y_u=Mv_u\), a contradiction.  Hence \(v_u\notin W_G\).
Appending
\[
        (v_u,y_{u0}+y_{u1})
        =
        (v_u,Mv_u)
\]
therefore preserves certification and increases
the dimension of the self-correctable space by at least one.

Now suppose that \(u\) is a leaf.  By the definition of leaf badness,
\[
        y_u\ne Mv_u \enspace.
\]
If \(v_u\in W_G\), then the absence of
\(\calE_{\mathrm{list}}\) and \cref{lem:select} would again imply
\(y_u=Mv_u\), a contradiction.  Thus $v_u\notin W_G$,
so the precondition for \(\Extract(v_u,G)\) holds.

We now verify that \(\Extract\) appends a certified new pair. During this call, keep \(G\) fixed as the list at the start of the call, and maintain the invariant that the current sparse vector \(v'\) satisfies \(v'\notin W_G\). This holds initially by the previous paragraph. If \(\wt(v')=1\), the procedure computes \(Mv'\) directly and appends \((v',Mv')\), which is certified and outside the current \(W_G\). Otherwise, write \(v'=v_0+v_1\), compute selected values \(y_0,y_1\), and run the two restricted verification tests. If the verifier rejects for a child \(v_i\), then by the one-sidedness of the verifier we have \(y_i\ne Mv_i\). If \(v_i\in W_G\), then \(\calE_{\mathrm{list}}\) and \cref{lem:select} would imply \(y_i=Mv_i\), a contradiction. Hence \(v_i\notin W_G\), and the recursive call preserves the invariant. If both child verifications accept, then by \(\calE_{\mathrm{ver}}\) both claims are true. Namely, $y_0=Mv_0$ and $y_1=Mv_1$. Thus, \[ y_0+y_1=Mv_0+Mv_1=Mv' \enspace . \] Since the invariant gives \(v'\notin W_G\), appending \[ (v',y_0+y_1)=(v',Mv') \enspace \] keeps \(G\) certified and makes genuine progress. Since the support size is halved at every recursive step, \(\Extract\) terminates after at most \(H=O(\log b)\) levels. We have shown that every iteration which does not output \(Mx\) appends one certified pair whose first component lies outside the current \(W_G\). Hence the quotient dimension of \(W_G\) increases by at least one in every non-terminating iteration. Since \(V\) has codimension \(O(\log(1/\eps))\), this can happen for at most \(O(\log(1/\eps))\) iterations. Our choice of \(L\) is larger than this bound. Therefore, conditioned on \(\calE\), the algorithm must eventually reach an iteration in which no bad node is found, and in that iteration it outputs \(Mx\).
Consequently,
\[
        \Pr\left[
            \Red_{b,\kappa}^{\OO,M}(x)=Mx
        \right]
        \ge
        \Pr[\calE]
        >
        \frac{2}{3}.
\]
\subsubsection*{Query and running-time bounds}
Let $K=q^{|G|}R$ be the candidate-list size in a given iteration.  Since the algorithm runs at most $O(\log(1/\eps))$ iterations, and the codimension of $V$ is bounded by this quantity, we always have $|G|=O(\log(1/\eps))$ on the good event.  Also $R=\ot(1/\eps^{2t+1})$ and $t=O(\log(1/\eps))$.  Therefore $K=\ot((q/\eps)^{O(\log(1/\eps))})$.

For fixed $G$ and fixed $\Omega$, one local bad-node predicate evaluation costs $\ot(KT+\Psi(n,b,\kappa))$.  Here the $KT$ term comes from list generation and tournament selection, while $\Psi(n,b,\kappa)$ is the block RS fingerprinted verification cost on a leaf.  In terms of queries, one predicate evaluation uses $\ot(K)$ queries to $U_{\OO}$ and $\ot(Kn+b\sqrt{n\kappa})$ queries to $U_M$.

The outer quantum search ranges over $N_\calT=O(n/b)$ tree nodes.  The additional descent to a minimal bad node has the same asymptotic cost, because the searched subtree sizes decrease geometrically.  Thus one iteration has running time
\begin{align*}
        \ot\left((KT+\Psi(n,b,\kappa))\sqrt{\frac nb}\right)
        \enspace .
\end{align*}
Using $T\ge n$, this is
\begin{align*}
        \ot\left(
        KT\sqrt{\frac nb}
        +
        \left(\frac n\kappa\right)^{3/2}\sqrt{\frac nb}
        +
        n\sqrt{b\kappa}
        \right)
        \enspace .
\end{align*}
The corresponding query bounds per iteration are $\ot(K\sqrt{n/b})$ queries to $U_{\OO}$ and
\begin{align*}
        \ot\left(K\frac{n^{3/2}}{\sqrt b}+n\sqrt{b\kappa}\right)
        \enspace
\end{align*}
queries to $U_M$.

The extraction cost is lower order in this bound.  Extraction has depth $O(\log b)$, and each level works with vectors of weight at most $b$.  Up to logarithmic factors, each level costs at most $\ot(KT+\Psi(n,b,\kappa))$, and the base case costs only one column read from $M$.  Since $\sqrt{n/b}\ge1$, this is absorbed by the iteration-search cost.  The final root selection costs $\ot(KT)$ and is absorbed as well.

Multiplying by the $O(\log(1/\eps))$ iterations and absorbing this factor into the global dependence gives the bounds in \cref{thm:main}.

\subsubsection*{Proof of \cref{thm:main} and \cref{cor:optimized}}

The correctness argument above gives success probability at least $2/3$.  Standard repetition can amplify this if desired.  The query and running-time bounds above are exactly the bounds stated in \cref{thm:main}.  This proves the theorem.

\begin{proof}[Proof of \cref{cor:optimized}]
For fixed $n$ and $T$, the $b,\kappa$-dependent part of the running time is
\begin{align*}
        (q/\eps)^{O(\log(1/\eps))}T\sqrt{\frac nb}
        +
        \left(\frac n\kappa\right)^{3/2}\sqrt{\frac nb}
        +
        n\sqrt{b\kappa}
        \enspace .
\end{align*}
When $n\le T\le n^{3/2}$, choose $b=\Theta(T^{4/3}/n)$ and $\kappa=\Theta(n/T^{2/3})$.  Then
\begin{align*}
        T\sqrt{\frac nb}=\Theta(nT^{1/3})
        \enspace ,
        \qquad
        \left(\frac n\kappa\right)^{3/2}\sqrt{\frac nb}=\Theta(nT^{1/3})
        \enspace ,
        \qquad
        n\sqrt{b\kappa}=\Theta(nT^{1/3})
        \enspace .
\end{align*}
This gives the claimed running time.  Substituting the same parameter values into the query bounds gives $Q_{\OO}(b,\kappa)=\ot((q/\eps)^{O(\log(1/\eps))}n/T^{2/3})$ and
\begin{align*}
        Q_M(b,\kappa)
        =
        \ot\left(
        (q/\eps)^{O(\log(1/\eps))}\frac{n^2}{T^{2/3}}
        +
        nT^{1/3}
        \right)
        \enspace .
\end{align*}
Since $T\ge n$, the first term is at most $\ot((q/\eps)^{O(\log(1/\eps))}nT^{1/3})$, so the simpler bound stated in the corollary follows.

If $T>n^{3/2}$, choose $b=n$ and $\kappa=1$.  Then the running-time expression is $\ot((q/\eps)^{O(\log(1/\eps))}T+n^{3/2})$, which is $\ot((q/\eps)^{O(\log(1/\eps))}T)$.  The query bounds reduce to the stated expressions.  In particular, when $T=\Theta(n)$, choosing $b=\kappa=\Theta(n^{1/3})$ gives the $\ot(n^{4/3})$ bound, up to the dependence on $q$ and $\eps$.
\end{proof}


\printbibliography

@inproceedings{AGGSS24,
  author    = {Asadi, Vahid R. and Golovnev, Alexander and Gur, Tom and Shinkar, Igor and Subramanian, Sathyawageeswar},
  title     = {Quantum worst-case to average-case reductions for all linear problems},
  booktitle = {Proceedings of the 2024 Annual ACM-SIAM Symposium on Discrete Algorithms (SODA)},

  year      = {2024},
  doi       = {10.1137/1.9781611977912.90},
}

@inproceedings{HS26,
author = {Hirahara, Shuichi and Shimizu, Nobutaka},
title = {Optimal Random Self-Reductions for All Linear Problems},
year = {2026},
url = {https://doi.org/10.1145/3798129.3800773},
doi = {10.1145/3798129.3800773},
booktitle = {Proceedings of the 58th Annual ACM Symposium on Theory of Computing},
  series    = {STOC 2026},

}

@inproceedings{AGGS22,
  author    = {Asadi, Vahid R. and Golovnev, Alexander and Gur, Tom and Shinkar, Igor},
  title     = {Worst-Case to Average-Case Reductions via Additive Combinatorics},
  booktitle = {Proceedings of the 54th Annual ACM SIGACT Symposium on Theory of Computing},
  series    = {STOC 2022},
  year      = {2022},
  doi       = {10.1145/3519935.3520041},
}

@book{NS02,
  title={Quantum computation and quantum information},
  author={Nielsen, Michael A and Chuang, Isaac L},
  year={2010},
  publisher={Cambridge University Press}
}

@article{BLR93,
  author  = {Manuel Blum and Michael Luby and Ronitt Rubinfeld},
  title   = {Self-Testing/Correcting with Applications to Numerical Problems},
  journal = {Journal of Computer and System Sciences},
  year    = {1993},
  doi     = {10.1016/0022-0000(93)90044-W}
}

@InProceedings{Gilyen2019,
author = {Gily{\'{e}}n, Andr{\'{a}}s and Su, Yuan and Low, Guang Hao and Wiebe, Nathan},
doi = {10.1145/3313276.3316366},
booktitle = {STOC 2019},
title = {{Quantum singular value transformation and beyond: Exponential improvements for quantum matrix arithmetics}},
year = {2019},
}

@inproceedings{HS22,
  author    = {Shuichi Hirahara and Nobutaka Shimizu},
  title     = {Hardness Self-Amplification from Feasible Hard-Core Sets},
  booktitle = {Proceedings of the 63rd IEEE Annual Symposium on Foundations of Computer Science},
  year      = {2022},
  doi       = {10.1109/FOCS54457.2022.00058}
}

@inproceedings{HS23,
  author    = {Shuichi Hirahara and Nobutaka Shimizu},
  title     = {Hardness Self-Amplification: Simplified, Optimized, and Unified},
  booktitle = {Proceedings of the 55th Annual ACM Symposium on Theory of Computing},
  year      = {2023},
  doi       = {10.1145/3564246.3585189}
}

@ARTICLE{GL89,
  title        = {A hard-core predicate for all one-way functions},
  author       = {Goldreich, Oded and Levin, Leonid A.},
  journaltitle = {Proceedings of Symposium on Theory of Computing (STOC)},
  publisher    = {ACM Press},
  location     = {New York, New York, USA},
  pages        = {25--32},
  date         = {1989},
  doi          = {10.1145/73007.73010},
  url          = {http://dx.doi.org/10.1145/73007.73010}
}

@INBOOK{Kal91,
  title     = {On {Wiedemann's} method of solving sparse linear systems},
  author    = {Kaltofen, Erich and David Saunders, B},
  booktitle = {Applied Algebra, Algebraic Algorithms and Error-Correcting Codes},
  publisher = {Springer Berlin Heidelberg},
  location  = {Berlin, Heidelberg},
  pages     = {29--38},
  date      = {1991},
  doi       = {10.1007/3-540-54522-0_93},
  series    = {Lecture notes in computer science},
  url       = {http://dx.doi.org/10.1007/3-540-54522-0_93},
  urldate   = {2025-10-22},
}

@misc{AK25,
      title={Quantum Worst-Case to Average-Case Reduction for Matrix-Vector Multiplication}, 
      author={Divesh Aggarwal and Dexter Kwan},
      year={2025},
      eprint={2510.15721},
      archivePrefix={arXiv},
      primaryClass={quant-ph},
      url={https://arxiv.org/abs/2510.15721}, 
}

\appendix
\crefalias{section}{appendix}


\section{{Proof of the warm-up algorithm}}\label{apx:proof}
Here we provide a detailed proof of the error-less warm-up in \cref{sec:warm-up-errorless}.  The proof works over any finite field $\F=\F_q$.
We borrow notation from \cref{sec:Certified directions and candidate lists}.

\begin{proof}[Proof of \cref{thm:warm-up-errorless}]
Let $L=O(\log(1/\eps))$ be a fixed number of rounds larger than the codimension of the BR subspace $V$.  Each repair path has depth $O(\log n)$, so there is a fixed upper bound $Q=O(L\log n)$ on the number of correction calls.  Choose the number of BR trials to be
\begin{align*}
        R=\left\lceil\frac{C}{\eps^{2t+1}}\log(Q+2)\right\rceil
        \enspace ,
\end{align*}
for a sufficiently large constant $C$.

For a fixed certified list $G$ and a fixed vector $v\in W_G$, choose coefficients $\alpha$ such that $v-w_\alpha\in V$.  By \cref{thm:BR}, each trial for this choice uses only inputs in $A$ with probability at least $\eps^{2t+1}$.  Such a trial gives the correct answer $Mv$.  Hence
\begin{align*}
        \Pr_\Omega[\mathsf{Correct}_{G,\Omega}(v)=\bot]
        \le(1-\eps^{2t+1})^R
        \enspace .
\end{align*}
With the choice above, a union bound over the at most $Q$ calls bounds the probability of any such failure by $1/10$.

Every non-$\bot$ output of $\mathsf{Correct}_{G,\Omega}$ is correct whenever $G$ is certified.  Indeed, as discussed above, all oracle values used by a good BR candidate are correct, and since $G$ is certified, any linear combination of its elements will be correct as well.  Note that on every execution, even if a list failure occurs, $G$ remains certified. Therefore, the algorithm can only output the correct product or $\bot$.

Condition on the event that no needed correction call returns $\bot$ for a vector in the current $W_G$.  If a round returns a vector, it is $Mx$.  Otherwise, correction of $x$ returned $\bot$, and therefore, because of the event we conditioned on, we have that $x\notin W_G$.  This gives the precondition of $\textsc{Repair}$.

We next show that $\textsc{Repair}$ indeed makes progress.  Its invariant is that the current vector $v$ is outside the space $W_G$ from the start of this call.  If $\wt(v)\le1$, the procedure computes $Mv$ directly and appends the certified pair $(v,Mv)$.  So this adds a direction outside $W_G$.

Otherwise, write $v=v_0+v_1$.  If correction of a child $v_i$ returns $\bot$, then $v_i\notin W_G$, and the recursive call preserves the invariant.  If both children return answers, then they are $Mv_0$ and $Mv_1$.  Their sum is $Mv$, so appending $(v,y_0+y_1)$ is certified and increases the dimension of $W_G$ by one.

Thus, every non-terminating round adds one certified direction outside the current $W_G$.  Since $V$ has codimension $O(\log(1/\eps))$, this can happen only $O(\log(1/\eps))$ times.  Our choice of $L$ leaves at least one final correction attempt after the self-correction space has become $\F^n$. The algorithm therefore returns $Mx$ with probability at least $9/10$, which is more than required.

Finally, throughout the execution, $|G|\le L$ and at most $q^L R$ BR candidates are tried in one correction call.  Each candidate uses $2t$ oracle calls and $O(tn)$ additional field operations.  Each repair call reads at most one column of $M$.  The total cost is therefore
\begin{align*}
        O\bigl(Qq^L Rt(T+n)+Ln\bigr)
        =\ot\left((q/\eps)^{O(\log(1/\eps))}T\right)
        \enspace .
\end{align*}
The number of entry queries to $M$ is $O(Ln)$.  For fixed $q$ and $\eps$, the displayed choice of $R$ gives the more explicit bound $O(T\log n\log\log n)$.  The extra $\log\log n$ comes from the union bound over the repair path, and it is suppressed by the $\ot(\cdot)$ notation in the overview.
\end{proof}

\begin{proof}[Proof of \cref{thm:warm-up-verifier}]
Given an error-prone oracle $\OO$ and a perfect verifier $\calV_M$, define
\begin{align*}
        \OO_\bot(z):=
        \begin{cases}
            \OO(z),&\calV_M(z,\OO(z))=\Yes\enspace ,\\
            \bot,&\text{otherwise}\enspace .
        \end{cases}
\end{align*}
We evaluate $\OO(z)$ once and store its answer before verification.  This is an error-less oracle with the same agreement as $\OO$, and one call costs $O(T+S)$.  Applying \cref{thm:warm-up-errorless} proves the first claim.

For the quantum implementation, let $Q_{\OO}$ be the fixed upper bound on the number of calls to $\OO_\bot$ in this reduction.  Use the matrix-vector verifier of \cite{AGGSS24} with error at most $1/(10Q_{\OO})$ per call.  Its cost is $\ot(n^{3/2})$, including this error reduction.

Conditioning on the transcript before each verification call and taking a union bound, the probability that any verifier answer differs from the ideal exact answer is at most $1/10$.  Until such an error, the execution agrees with the one using the perfect verifier.  The latter fails with probability at most $1/10$ by the proof above.  Thus the total failure probability is at most $1/5<1/3$.  Substituting the verification cost gives the claimed running time. 
\end{proof}


\section{{Randomized and quantum average-case oracles}}\label{apx:randomized-oracles}
We stated the main reduction for deterministic oracles for simplicity of exposition.  Here we extend it to randomized and quantum oracles.  The main technical idea to extend the result to this setting is to define a new oracle, which runs the original non-deterministic oracle several times and applies tournament selection on the list of candidates.  This gives a dense set of inputs on which the new oracle is correct with high probability on every fresh run.  We use this set in the BR argument, and account for fresh answers during the consistency-tree search.

As in the main reduction, we assume a coherent implementation of the oracle and access to its inverse.  This includes randomized and quantum algorithms with their measurements deferred.  We can formalize the oracle's action on a clean workspace as

\begin{align*}
        U_{\OO}|z\rangle|0\rangle|0\rangle
        =|z\rangle\sum_{y\in\F^n}
        \alpha_{z,y}|y\rangle|\gamma_{z,y}\rangle
        \enspace ,
\end{align*}
where the second register holds the output, the third contains the oracle's workspace,
and $\sum_y|\alpha_{z,y}|^2=1$. Measuring the output register gives
$y$ with probability $|\alpha_{z,y}|^2$. Thus,
$p_z:=|\alpha_{z,Mz}|^2$ is the probability of the correct answer.
The average-case assumption is
\begin{align}\label{eq:randomized-agreement}
        \E_{z\sim\F^n}[p_z]\ge\eps
        \enspace .
\end{align}
We now state the main result of this appendix.

\begin{theorem}\label{thm:randomized-oracles}
The conclusions of \cref{thm:main} also hold for an oracle satisfying \cref{eq:randomized-agreement}.  This includes the optimized bounds in \cref{cor:optimized}.  Here $T\ge n$ is the cost of one use of the original oracle or its inverse, and the success probability includes the oracle's internal randomness and quantum measurements.
\end{theorem}

The starting point is to adjust the definition of the set of good inputs to the new setting, by amplifying the success rate of the oracle.

\paragraph{Amplifying the oracle.}
Let $A_0:=\{z:p_z\ge\eps/2\}$.  If $\mu=|A_0|/q^n$, then
\begin{align*}
        \eps\le\E_z[p_z]\le\mu+(1-\mu)\frac{\eps}{2}
        \enspace ,
\end{align*}
so $\mu\ge\eps/(2-\eps)\ge\eps/2$.  For an error parameter $\delta$, define $\OO'(z)$ by running $\OO(z)$ independently
\begin{align*}
        k=\left\lceil\frac{2}{\eps}\ln\frac1\delta\right\rceil
\end{align*}
times and applying $\Select$ to the resulting list.  For $z\in A_0$, at least one run is correct except with probability at most $(1-p_z)^k\le\delta$.  By \cref{lem:select}, tournament selection then returns $Mz$.  Therefore the set
\begin{align}\label{eq:randomized-good-set}
        A':=\{z:\Prb[\OO'(z)=Mz]\ge1-\delta\}
\end{align}
has density at least $\eps/2$.

This construction also works in superposition.  For each input branch, we run the $k$ copies on separate work registers and apply the reversible tournament to their output registers.  Whenever a basis list contains the correct answer, selection returns it, so the same probability bound holds.  One use of $\OO'$ costs $T'=O(k(T+n))=O(kT)$ and uses at most $kn$ additional entry queries to $M$.

\paragraph{BR correction.}
Put $\eps_0:=\eps/2$, and apply \cref{thm:BR} to $A'$ with this density lower bound.  Let $V$ be the resulting subspace, and choose $t=O(\log(1/\eps_0))$ as in that lemma.  The BR tables contain sampled input vectors, just as in the deterministic reduction. 

For a fixed certified $G$ and $v\in W_G$, a trial for an appropriate coefficient choice has all its $2t$ terms in $A'$ with probability at least $\eps_0^{2t+1}$.  Thus, by taking $R=\ot(\eps_0^{-(2t+1)})$, with high probability the table contains at least one such trial for every needed vector in $W_G$.  The union bound over tree nodes and adaptively chosen extraction vectors is the same as in \cref{sec:full-algorithm}.

Fix a table with this property.  For any one of these vectors $v$, the $2t$ calls to $\OO'$ in its successful trial are all correct except with probability at most $\nu:=2t\delta$.  Hence, the selected value obtained from the full BR candidate list is correct with probability at least $1-\nu$ on every fresh evaluation.  For a tree node $u$, denote this randomized selected value by $Y_u$, and write
\begin{align*}
        e_u:=\Prb[Y_u\ne Mv_u]
        \enspace .
\end{align*}
On the good-table event, we have
\begin{align}\label{eq:randomized-node-good}
        v_u\in W_G\quad\Longrightarrow\quad e_u\le\nu
        \enspace .
\end{align}
The table and $G$ are fixed here, but all oracle evaluations use fresh workspace.  In particular, we do not assume that $Y_u$ has the same value on different runs.

\paragraph{Consistency tests with fresh answers.}
For vectors in $W_G$, BR correction already gives the correct answer with high probability.  However, the children of a tree node need not lie in $W_G$, but at the same time,  their BR answers may be correct with non-negligible probability, without being correct with high probability.  This matters when we repair at an internal node: we use the sum of its child values as its correct value, so we need both child values to be correct.

We address this by using a larger BR candidate list for each child when we check their consistency with their parents.  For any node $u$, keep $G$ and the BR table fixed, evaluate the BR candidate list $s$ times using fresh calls to $\OO'$, concatenate the resulting lists, and apply $\Select$ once to this larger list.  Denote the result by $\widehat Y_u$.  Thus, we use the same BR input vectors in all copies, but evaluate the oracle independently in each copy.

If any one of these copies would give the correct selected answer, then the larger list contains $Mv_u$, and tournament selection returns it.  Therefore
\begin{align}\label{eq:boosted-node}
        \Prb[\widehat Y_u\ne Mv_u]\le e_u^s
        \enspace .
\end{align}
In particular, if $e_u\le1/4$, then the error is at most $4^{-s}$, whether or not $v_u$ lies in $W_G$. Taking $s$ logarithmic in $n$ and the error parameters then makes the child errors sufficiently small.  Note that this only adds logarithmic factors to the cost of a node computation.

At an internal node $u$, compute fresh values $Y_u,\widehat Y_{u0},\widehat Y_{u1}$, and raise a flag if
\begin{align*}
        Y_u\ne\widehat Y_{u0}+\widehat Y_{u1}
        \enspace .
\end{align*}
The parent uses the original BR list, and each child uses the larger list.  At a leaf, compute $Y_u$ and use the restricted verifier with error at most $1/3$ to flag an incorrect product, and keep the generated values in the workspace.

Let $p_u$ be the flag probability at node $u$.  If the parent answer is incorrect and both child answers are correct, the internal test must flag.  At a leaf, the verifier detects an incorrect answer with probability at least $2/3$.  Hence
\begin{align*}
        e_u&\le p_u+e_{u0}^s+e_{u1}^s
        &&\text{if }u\text{ is internal}\enspace ,\\
        e_u&\le\frac32p_u
        &&\text{if }u\text{ is a leaf}\enspace .
\end{align*}
Since $e_w^s\le e_w$, induction on the tree gives
\begin{align}\label{eq:subtree-error-mass}
        e_u\le2\sum_{w\in\Nodes(\calT_u)}p_w
        \enspace ,
\end{align}
where $\calT_u$ is the subtree rooted at $u$.

For a set $S$ of nodes, prepare a uniform superposition over $S$ and apply this flagging circuit.  Its flag probability is $|S|^{-1}\sum_{u\in S}p_u$.  Apply \cref{lem:fpaa} with lower bound $1/(8|S|)$.  In time $\ot(\sqrt{|S|})$ times the cost of one flag computation, this search returns a flagged node with high probability whenever $\sum_{u\in S}p_u\ge1/8$.  If it returns a node, retain the values observed with the flag.  Outside its allotted failure event, a return of $\Fail$ implies
\begin{align*}
        \sum_{u\in S}p_u<\frac18
        \enspace .
\end{align*}
This use of amplitude amplification only needs a fixed coherent circuit and its inverse, and does not need a deterministically selected value at each node.

The search and descent are now as follows.  First search the full tree.  If the search returns $\Fail$, output a fresh $\widehat Y_r$ at the root.  Otherwise, let $u$ be the returned node, together with its observed values.  If $u$ is internal, search all its proper descendants.  If that search finds a flagged node, move to it and continue; if it returns $\Fail$, stop at $u$.  At this stopping point, append
\begin{align*}
        (v_u,\widehat Y_{u0}+\widehat Y_{u1})
\end{align*}
using the values retained from the search that returned $u$.  If the descent reaches a leaf, run $\Extract$ as before.  Each move goes down at least one level, so the searched subtree sizes decrease geometrically.  The total search cost is still $\ot(\sqrt{N_{\calT}})$ flag computations.

To see why the larger child lists are useful, suppose the descent stops at an internal node $u$ and the search guarantee holds.  The total flag probability over its proper descendants is less than $1/8$.  By \cref{eq:subtree-error-mass}, this gives $e_{u0},e_{u1}<1/4$.  It does not show that either child belongs to $W_G$.  Nevertheless, \cref{eq:boosted-node} bounds the error of each larger child list by $4^{-s}$ in a fresh computation.  We keep the values that caused the original disagreement, since recomputing them could give different answers.  The next paragraph accounts for the effect of quantum search on the probability that these retained values are incorrect.

\paragraph{Errors under amplitude amplification.}
The bounds above concern one fresh node computation.  Quantum search can increase the probability of rare flagged outcomes, including flags caused by incorrect child values.  We therefore need to bound these errors after the search has selected a node.

Condition on a certified $G$ and a good BR table.  There are two unwanted kinds of flagged outcomes:
\begin{enumerate}
    \item The node satisfies $v_u\in W_G$, but its observed value $Y_u$ is incorrect.
    \item The node is internal, $e_{u0},e_{u1}\le1/4$, but at least one observed child value from the larger BR lists is incorrect.
\end{enumerate}
The first event would prevent an incorrect parent answer from certifying that $v_u\notin W_G$.  The second would prevent us from using the child sum as the correct value for the parent.  At a stopping internal node, the descendant search gives the constant bounds on the child error probabilities.  If neither unwanted event occurs, the retained child values are correct, their disagreement with the parent shows that its observed answer is incorrect, and the parent vector lies outside $W_G$.

By \cref{eq:randomized-node-good,eq:boosted-node}, the total probability of these unwanted outcomes in any one unamplified node computation is at most $\nu+2\cdot4^{-s}$.  The same bound holds for a uniform superposition over any set of nodes.  Thus, enlarging the child lists makes the second kind of error small even for children outside $W_G$.

We use a simple consequence of \cref{lem:fpaa} to account for the search.  Suppose an unwanted part of the flagged state has probability at most $\zeta$ before amplification, and the amplification circuit uses $a$ calls to the preparation unitary and its inverse.  The probability of this unwanted outcome afterwards is $O(a^2\zeta)$.

To see this, rotate the prepared state to remove its unwanted component and normalize the rest.  This changes the preparation unitary in operator norm by $O(\sqrt\zeta)$.  By a hybrid argument over its $a$ uses, the amplified state changes in norm by $O(a\sqrt\zeta)$.  For the modified preparation, the unwanted flagged component is zero after amplification: by \cref{lem:fpaa}, the final flagged component is proportional to the initial one.  Squaring the norm bound proves the claim.  This rotation is only used for the analysis.

Here are choices that make the total error small.  Let $L=O(\log(1/\eps_0))$ bound the number of outer iterations and let $H=O(\log n)$ bound the tree and extraction depths.  Give each measured search failure probability at most $1/(CLH)$, where $C$ is a sufficiently large constant.  Let $P=\ot(L\sqrt n)$ be a fixed upper bound on the total number of uses of the flag-preparation unitaries and their inverses in these searches.  Choose
\begin{align*}
        \delta:=\frac{1}{Ct(P^2+LH+1)}
        \enspace ,
        \qquad
        s:=\left\lceil C\log(P+LH+2)\right\rceil
        \enspace .
\end{align*}
The sum of the squared numbers of preparation calls is at most $P^2$.  The total probability of an unwanted amplified outcome is thus at most $O(P^2(\nu+4^{-s}))$, which can be made smaller than $1/30$.  Also, $\delta$ is inverse polynomial in $n$ for fixed $\eps$, and $s$ is logarithmic.  These choices do not change the asymptotic bounds.

\paragraph{Correctness.}
First suppose the full-tree search returns $\Fail$ and its guarantee holds.  By \cref{eq:subtree-error-mass}, $e_r<1/4$.  A fresh amplified root computation therefore returns $Mx$ except with probability at most $4^{-s}$.

Next suppose the descent stops at an internal node $u$.  Its proper-descendant search returned $\Fail$, so \cref{eq:subtree-error-mass} gives $e_{u0},e_{u1}<1/4$.  If the second unwanted event does not occur, both retained child values are correct.  Since the retained flag records a disagreement, the retained parent value $Y_u$ is incorrect.  If the first unwanted event does not occur either, this implies $v_u\notin W_G$.  Thus the appended pair is exactly $(v_u,Mv_u)$ and adds a direction outside $W_G$.

At a flagged leaf, the one-sided verifier guarantees that the observed value is incorrect.  Excluding the first unwanted event again gives $v_u\notin W_G$.  The proof for $\Extract$ is unchanged, except that its measured child computations also have failure probability at most $\nu$ when their vectors are in $W_G$ and their fresh BR tables are good.  There are $O(LH)$ such computations, and our choice of $\delta$ makes their total additional failure probability small.

Choose the BR sample count as in the main proof, with $\eps_0$ in place of $\eps$, so that all needed tables are good except with probability at most $1/30$.  The remaining events are search failures, unwanted amplified outcomes, extraction-verification errors, and the fresh selected-value errors just described.  The choices above, together with the verification error reduction from the main proof, make their total probability less than $1/3$.  The bounds are conditional on the previous transcript whenever the current $G$ is certified, so they also hold for the adaptive execution.  Outside these events, every non-terminating iteration adds a certified direction outside $W_G$.  The same dimension argument as before proves termination with the correct answer.

\paragraph{Cost.}
The initial amplification uses $k=\ot(1/\eps)$ calls to $\OO$ and $O(kn)$ additional matrix queries per evaluation of $\OO'$.  The amplification of child values uses only $s=O(\log n+\log L)$ repetitions and tournament comparisons.  Therefore one node flag costs $\ot(kKT+\Psi(n,b,\kappa))$, and the additional matrix queries are $\ot(kKn)$.  Since $kK=\ot((q/\eps)^{O(\log(1/\eps))})$, multiplying by the unchanged $\ot(\sqrt{n/b})$ search cost and accounting for extraction gives exactly the asymptotic bounds in \cref{thm:main,cor:optimized}.  In particular, the running time remains $\ot(nT^{1/3})$ for constant field size and agreement when $n\le T\le n^{3/2}$.

\begin{corollary}[Quantum algorithms for linear problems]
\label{cor:quantum-algorithms}
Let $M\in\F^{n\times n}$, where $|\F|=q$, and let $0<\eps\le1/2$.  Assume coherent entry access to $M$, with each query counted as one oracle gate.  Suppose a quantum algorithm $\mathcal A$ has gate complexity $T\ge n$ and satisfies
\begin{align*}
        \E_{z\sim\F^n}
        \left[\Prb[\mathcal A(z)=Mz]\right]
        \ge\eps
        \enspace ,
\end{align*}
where the inner probability is over the algorithm's internal randomness and quantum measurements.  Then there is a quantum algorithm $\mathcal B$, with coherent entry access to $M$, such that
\begin{align*}
        \Prb[\mathcal B(x)=Mx]\ge\frac23
        \qquad\text{for every }x\in\F^n
        \enspace ,
\end{align*}
and whose gate complexity is
\begin{align*}
        \ot\left(
        (q/\eps)^{O(\log(1/\eps))}
        \left(T+nT^{1/3}\right)
        \right)
        \enspace .
\end{align*}
\end{corollary}

\begin{proof}
Defer the measurements of $\mathcal A$, retain its workspace, and preserve a copy of its input.  This gives a coherent implementation of the average-case oracle with the same output probabilities.  Its inverse is obtained by reversing the circuit and replacing each gate by its inverse.  Both computations have gate complexity $\ot(T)$.

We can therefore apply \cref{thm:randomized-oracles} and the parameter choices in \cref{cor:optimized}.  If $n\le T\le n^{3/2}$, the resulting gate complexity is
\begin{align*}
        \ot\left(
        (q/\eps)^{O(\log(1/\eps))}nT^{1/3}
        \right)
        \enspace .
\end{align*}
If $T>n^{3/2}$, choosing $b=n$ and $\kappa=1$ gives gate complexity
\begin{align*}
        \ot\left(
        (q/\eps)^{O(\log(1/\eps))}T+n^{3/2}
        \right)
        =
        \ot\left(
        (q/\eps)^{O(\log(1/\eps))}T
        \right)
        \enspace .
\end{align*}
Combining the two cases gives the stated bound.  Uniformity follows from the uniformity of the reduction and the circuit construction above.
\end{proof}

\end{document}